\documentclass[11pt]{article}
\usepackage{ifthen}
\usepackage{mdwlist}
\usepackage{amsmath,amssymb,amsfonts,amsthm}
\usepackage{bm}
\usepackage[colorlinks=true,linkcolor=blue,citecolor=blue,urlcolor=blue]{hyperref}
\usepackage[utf8]{inputenc}
\usepackage{enumitem}
\usepackage{graphicx}
\usepackage{xspace}
\usepackage{verbatim}
\usepackage{algorithm}
\usepackage{algpseudocode}
\usepackage[margin=1in]{geometry}
\usepackage{color}
\usepackage{thmtools}
\usepackage{thm-restate}
\usepackage{multicol}
\usepackage{multirow}
\usepackage[T1]{fontenc}
\usepackage{latexsym}
\usepackage{epsfig}
\usepackage{bbm}
\graphicspath{{plots/}}

{}

\def\b{\beta}

\def\d{\delta}

\renewcommand{\epsilon}{\ve}
\def\ve{\varepsilon}

\def\g{\gamma}
\def\G{\Gamma}

\def\z{\zeta}
\def\th{\theta}

\def\l{\lambda}

\def\m{\mu}

\def\p{\pi}

\def\s{\sigma}

\def\t{\tau}

\newcommand{\E}{\mbox{\bf E}}

\newcommand{\Var}{\mbox{\bf Var}}
\newcommand{\sign}{\mbox{\bf sign}}

\newcommand{\pr}[2][]{\mbox{Pr}\ifthenelse{\not\equal{}{#1}}{_{#1}}{}\!\left[#2\right]}

\newcommand{\reals}{\mathbb{R}}

\newcommand{\U}{\mathcal{U}_n}
\newcommand{\I}{\mathcal{I}_n}

\newcommand{\dtv}{d_{\mathrm {TV}}}

\newcommand{\dkl}{d_{\mathrm {KL}}}
\newcommand{\dskl}{d_{\mathrm {SKL}}}

\newcommand{\abss}[1]{\left\lvert {#1} \right\rvert}

\newtheorem{theorem}{Theorem}

\newtheorem{proposition}{Proposition}

\newtheorem{remark}{Remark}

\newtheorem{lemma}{Lemma}
\newtheorem{claim}{Claim}
\newtheorem{corollary}{Corollary}

\newtheorem{definition}{Definition}
\newtheorem{question}{Question}

\newcommand{\ignore}[1]{}
\providecommand{\poly}{\operatorname*{poly}}

\newenvironment{prevproof}[2]{\noindent {\em {Proof of {#1}~\ref{#2}:}}}{$\hfill\qed$\vskip \belowdisplayskip}

\newcommand{\bg}[1]{\medskip\noindent{\bf #1}}

\definecolor{Red}{rgb}{1,0,0}

\newcommand{\oldbound}[1]{{}}

\newcommand{\goodgamma}{\chi_{\tau_2}}
\newcommand{\dm}{d_{\max}}

\DeclareMathOperator{\sech}{sech}

\title{Testing Ising Models\footnote{A preliminary version of this paper appeared in the Proceedings of the 29th Annual ACM-SIAM Symposium on Discrete Algorithms (SODA 2018).}}

\author {
  Constantinos Daskalakis\thanks{EECS, MIT. {\tt costis@mit.edu}. This work was performed when supported by a Microsoft Research Faculty Fellowship, NSF CCF-1551875, CCF-1617730, CCF-1650733, and ONR N00014-12-1-0999.}
\and
Nishanth Dikkala\thanks{EECS, MIT. {\tt nishantd@csail.mit.edu}. This work was performed when supported by NSF CCF-1551875, CCF-1617730, CCF-1650733, and ONR N00014-12-1-0999.}
\and
Gautam Kamath\thanks{Cheriton School of Computer Science, University of Waterloo. {\tt g@csail.mit.edu}. This work was performed when author was a graduate student in EECS at MIT, and supported by NSF CCF-1551875, CCF-1617730, CCF-1650733, and ONR N00014-12-1-0999.} 
}
\begin{document}
\addtocounter{page}{-2}
\maketitle
\thispagestyle{empty}
\begin{abstract}

Given samples from an unknown multivariate distribution $p$, is it possible to distinguish whether $p$ is the product of its marginals versus $p$ being far from every product distribution? Similarly, is it possible to distinguish whether $p$ equals a given distribution $q$ versus $p$ and $q$ being far from each other? These problems of testing independence and goodness-of-fit have received enormous attention in statistics, information theory, and theoretical computer science, with sample-optimal algorithms known in several interesting regimes of parameters. Unfortunately, it has also been understood that these problems become intractable in large dimensions, necessitating exponential sample complexity.

Motivated by the exponential lower bounds for general distributions as well as the ubiquity of Markov Random Fields (MRFs) in the modeling of high-dimensional distributions, we initiate the study of distribution testing on \emph{structured} multivariate distributions, and in particular the prototypical example of MRFs: {\em the Ising Model}. We demonstrate that, in this structured setting, we can avoid the curse of dimensionality, obtaining sample and time efficient testers for independence and goodness-of-fit. One of the key technical challenges we face along the way is bounding the variance of functions of the Ising model.
\end{abstract}

\newpage
\thispagestyle{empty}
\tableofcontents
\thispagestyle{empty}
\newpage

\section{Introduction}
Two of the most classical problems in Statistics are testing independence and goodness-of-fit. {\em Independence testing} is the problem of distinguishing, given samples from a multivariate distribution $p$, whether or not it is the product of its marginals. The applications of this problem abound: for example, a central problem in genetics is to test, given genomes of several individuals, whether certain single-nucleotide-polymorphisms (SNPs) are independent from each other. In anthropological studies, a question that arises over and over again is testing whether the behaviors of individuals on a social network are independent; see e.g.~\cite{ChristakisF07}. The related problem of {\em goodness-of-fit testing} is that of distinguishing, given samples from $p$, whether or not it equals a specific ``model'' $q$. This problem arises whenever one has a  hypothesis (model) about the random source generating the samples and needs to verify whether the samples conform to the hypothesis.

Testing independence and goodness-of-fit have a long history in statistics, since the early days; for some old and some more recent references see, e.g.,~\cite{Pearson00,Fisher35,RaoS81,Agresti11}. Traditionally, the emphasis has been on the asymptotic analysis of tests, pinning down their error exponents as the number of samples tends to infinity~\cite{Agresti11,TanAW10}. In the two decades or so, distribution testing  has also piqued the interest of  theoretical computer scientists, where the emphasis has been different. In contrast to much of the statistics literature, the goal has been to obtain finite sample bounds for these problems. From this vantage point, our testing problems take the following form:

\vspace{5pt}\framebox{\begin{minipage}[h]{15.3cm} {\em Goodness-of-fit (or Identity) Testing:} {Given sample access to an unknown distribution $p$ over $\Sigma^n$ and a parameter $\epsilon>0$, the goal is to distinguish with probability at least $2/3$ between $p = q$ and $d(p,q) > \epsilon$, for some specific distribution $q$, from as few samples as possible.}

\medskip {\em Independence Testing:} {Given sample access to an unknown distribution $p$ over $\Sigma^n$ and a parameter $\epsilon>0$, the goal is to distinguish with probability at least $2/3$ between $p \in {\cal I}(\Sigma^n)$ and $d(p, {\cal I}(\Sigma^n)) > \epsilon$, where ${\cal I}(\Sigma^n)$ is the set of product distributions over $\Sigma^n$, from as few samples as possible.}
\end{minipage}
} 

\medskip \noindent In these problem definitions, $\Sigma$ is some discrete alphabet, and $d(\cdot,\cdot)$ some notion of distance or divergence between distributions, such as the total variation distance or the KL divergence. As usual, $2 \over 3$ is an arbitrary choice of a constant, except that it is bounded away from $1 \over 2$. It can always be boosted to some arbitrary $1-\delta$ at the expense of a multiplicative factor of $O(\log 1/\delta)$ in the sample complexity.

For both testing problems, recent work has identified tight upper and lower bounds on their sample complexity~\cite{Paninski08,ValiantV17a,AcharyaDK15,DiakonikolasK16}: when $d$ is taken to be the total variation distance, the optimal sample complexity for both problems turns out to be $\Theta\left({|\Sigma|^{n/2} \over \epsilon^2}\right)$, i.e. exponential in the dimension. As modern applications commonly involve high-dimensional data, this curse of dimensionality makes the above testing goals practically unattainable. Nevertheless, there {\em is} a sliver of hope, and it lies with the nature of all known sample-complexity lower bounds, which construct highly-correlated distributions that are hard to distinguish from the set of independent distributions~\cite{AcharyaDK15, DiakonikolasK16}, or from a particular distribution $q$~\cite{Paninski08}. Worst-case analysis of this sort seems overly pessimistic, as these instances are unlikely to arise in real-world data. As such, we propose testing high-dimensional distributions which are \emph{structured}, and thus could potentially rule out such adversarial distributions.

Motivated by the above considerations and the ubiquity of Markov Random Fields (MRFs) in the modeling of high-dimensional distributions (see~\cite{Jordan10} for the basics of MRFs and the references \cite{SanghaviTW10,KoNS07} for a sample of applications), we initiate the study of distribution testing for the prototypical example of MRFs: {\em the Ising Model,} which captures all binary MRFs with node and edge potentials.\footnote{This follows trivially by the definition of MRFs, and elementary Fourier analysis of Boolean functions.} Recall that the Ising model is a distribution over $\{-1,1\}^n$, defined in terms of a graph $G=(V,E)$ with $n$ nodes. It is parameterized by a scalar parameter $\theta_{u,v}$ for every edge $(u,v) \in E$, and a scalar parameter $\theta_{v}$ for every node $v \in V$, in terms of which it samples a vector $x \in \{\pm 1\}^V$ with probability:
\begin{align}
p(x) = {\rm exp} \left(\sum_{v\in V} \theta_v x_v + \sum_{(u,v) \in E} \theta_{u,v} x_u x_v - \Phi(\vec \th) \right), \label{eq:ising model}
\end{align}
where $\vec \th$ is the parameter vector and $\Phi(\vec \th)$ is the log-partition function, ensuring that the distribution is normalized. Intuitively, there is a random variable $X_v$ sitting on every node of $G$, which may be in one of two states, or spins: up (+1) or down (-1). The scalar parameter $\theta_v$ models a local (or ``external'') field at node $v$. The sign of $\theta_v$ represents whether this local field favors $X_v$ taking the value $+1$, i.e. the up spin, when $\theta_v>0$, or the value $-1$, i.e. the down spin, when $\theta_v<0$, and its magnitude represents the strength of the local field. We will say a model is ``without external field'' when $\th_v = 0$ for all $v \in V$. Similarly, $\theta_{u,v}$ represents the direct interaction between nodes $u$ and $v$. Its sign represents whether it favors equal spins, when $\theta_{u,v}>0$, or opposite spins, when $\theta_{u,v}<0$, and its magnitude corresponds to the strength of the direct interaction. Of course, depending on the structure of the Ising model and the edge parameters, there may be indirect interactions between nodes, which may overwhelm local fields or direct interactions.

The Ising model has a rich history, starting with its introduction by statistical physicists as a probabilistic model to study phase transitions in spin systems~\cite{Ising25}. Since then it has found a myriad of applications in diverse research disciplines, including probability theory, Markov chain Monte Carlo, computer vision, theoretical computer science, social network analysis, game theory, and computational biology~\cite{LevinPW09,Chatterjee05,Felsenstein04,DaskalakisMR11,GemanG86,Ellison93,MontanariS10}. 
The ubiquity of these applications motivate the problem of inferring Ising models from samples, or inferring statistical properties of Ising models from samples. 
This type of problem has enjoyed much study in statistics, machine learning, and information theory, see discussion in the related work. Much of prior work has focused on \emph{parameter learning}, where the goal is to determine the parameters of an Ising model to which sample access is given.
In contrast to this type of work, which focuses on discerning \emph{parametrically} distant Ising models, our goal is to discern \emph{statistically} distant Ising models, in the hopes of dramatic improvements in the sample complexity. 
To be precise, we study the following problems:


\vspace{5pt}\framebox{
\noindent \begin{minipage}[h]{15.3cm} {\em Ising Model Goodness-of-fit (or Identity) Testing:} {Given sample access to an unknown Ising model $p$ (with unknown parameters over an unknown graph) and a parameter $\epsilon>0$, the goal is to distinguish with probability at least $2/3$ between $p = q$ and $\dskl(p,q) > \epsilon$, for some specific Ising model $q$, from as few samples as possible.}	

\medskip {\em Ising Model Independence Testing:} {Given sample access to an unknown Ising model $p$ (with unknown parameters over an unknown graph) and a parameter $\epsilon>0$, the goal is to distinguish with probability at least $2/3$ between $p \in \I$ and $\dskl(p, \I) > \epsilon$, where $\I$ are all product distributions over $\{-1,1\}^n$, from as few samples as possible.}
\end{minipage}
}

\smallskip \noindent We note that there are several potential notions of statistical distance one could consider --- classically, total variation distance and the Kullback-Leibler (KL) divergence have seen the most study. As our focus here is on upper bounds, we consider the symmetrized KL divergence $\dskl$, which is a ``harder'' notion of distance than both: in particular, testers for $\dskl$ immediately imply testers for both total variation distance and the KL divergence.
 Moreover, by virtue of the fact that $\dskl$ upper-bounds KL in both directions, our tests offer useful information-theoretic interpretations of rejecting a model $q$, such as data differencing and large deviation bounds in both directions. 

\subsection{Results and Techniques}
Our main result is the following:

\begin{theorem}[Informal] \label{thm:meta-upper bound} Both Ising Model Goodness-of-fit Testing and Ising Model Independence Testing can be solved from ${\rm poly}\left(n, {1 \over \epsilon}\right)$ samples in polynomial time.
\end{theorem}
\noindent There are several variants of our testing problems, resulting from different knowledge that the analyst may have about the structure of the graph (connectivity, density), the nature of the interactions (attracting, repulsing, or mixed), as well as the temperature (low versus high). 
We proceed to discuss our results and techniques for a number of variants, allowing us to instantiate the resulting polynomial sample complexity in the above theorem.

\paragraph{A Baseline: Testing via Localization.}
In the least favorable setting when the analyst is oblivious to the structure of the Ising model $p$, the signs of the interactions, and their strength, the polynomial in Theorem~\ref{thm:meta-upper bound} becomes $O\left({n^4 \beta^2 + n^2 h^2 \over \epsilon^2}\right)$.
In this expression, $\beta=\max\{|\theta^p_{u,v}|\}$ for independence testing, and $\beta=\max\{|\theta^p_{u,v}|,|\theta^q_{u,v}|\}$ for goodness-of-fit testing, while $h = 0$ for independence testing, and $h = \max\{|\theta^p_u|,|\theta^q_u|\}$ for goodness-of-fit testing.
If the analyst has an upper bound on the maximum degree $\dm$, the dependence improves to $O\left({n^2 \dm^2 \beta^2 + n \dm h^2 \over \epsilon^2}\right)$, while if the analyst has an upper bound on the total number of edges $m$, then $\max\{m,n\}$ takes the role of $n \dm$ in the previous bound.
These results are summarized in Theorem~\ref{thm:localization}. 

These bounds are obtained via a simple localization argument showing that, whenever two Ising models $p$ and $q$ satisfy $\dskl(p,q)>\epsilon$, then ``we can blame it on a node or an edge;'' i.e.  there exists a node with significantly different bias under $p$ and $q$ or a pair of nodes $u, v$ whose {covariance} is significantly different under the two models. 
Pairwise correlation tests are a simple screening that is often employed in practice. For our setting, there is a straighforward and elegant way to show that pair-wise (and not higher-order) correlation tests suffice; see Lemma~\ref{lem:localization-skl}.
For more details about our baseline localization tester see Section~\ref{sec:localization}.


\paragraph{Exploiting Structure: Strong Localization for Trees and Ferromagnets.}
When $p$ is a tree-structured (or forest-structured) Ising model, then independence testing can be performed computationally efficiently without any dependence on $\beta$, with an additional quadratic improvement with respect to the other parameters. 
In particular, without external fields (i.e. $\max\{|\theta_u^p|\}=0$), independence can be testing with $O({n \over \epsilon})$ samples, and this result is tight when $m = O(n)$; see Theorem~\ref{thm:forests-independence} for an upper bound and Theorem~\ref{thm:linear-lb} for a lower bound.
Interestingly, we show the dependence on $\beta$ cannot be avoided in the presence of external fields, or if we switch to the problem of identity testing; see Theorem~\ref{thm:bh-lb}. 
In the latter case of identity testing, we can at least maintain the linear dependence on $n$; see Theorem~\ref{thm:forests-identity}.
Similar results hold when $p$ is a ferromagnet, i.e. $\theta^p_{u,v} \ge 0$, with no external fields, even if it is not a tree.
In particular, the sample complexity becomes $O({\max\{m,n\} \over \epsilon})$ (which is again tight when $m = O(n)$), see Theorem~\ref{thm:ferro-independence}.

These improvements are all obtained via the same localization approach discussed earlier, which resulted into our baseline tester. 
That is, we are still going to ``blame it on a node or an edge.'' 
The removal of the $\beta$ dependence and the improved running times are due to the proof of a structural lemma, which relates the parameter $\theta_{u,v}$ on some edge $(u,v)$ of the Ising model to the $\E[X_u X_v]$. 
We show that for forest-structured Ising models with no external fields, $\E[X_u X_v] = \tanh(\theta_{u,v})$, see Lemma~\ref{lem:trees-structural}.
A similar statement holds for ferromagnets with no external field, i.e., $\E[X_u X_v] \ge \tanh(\theta_{u,v})$, see Lemma~\ref{lem:ferro-structural}.
The proof of the structural lemma for trees/forests is straightforward. 
Intuitively, the only source of correlation between the endpoints $u$ and $v$ of some edge $(u,v)$ of the Ising model is the edge itself, as besides this edge there are no other paths between $u$ and $v$ that would provide alternative avenues for correlation. 
The proof of the inequality for ferromagnets on arbitrary graphs requires some work and proceeds by using Griffiths inequality for ferromagnets.
Further details about our algorithms for trees and ferromagnets can be found in Sections~\ref{sec:forests} and~\ref{sec:ferro}, respectively.

\paragraph{High Temperature and Dobrushin's Condition: Low-Variance Global Statistics.}
Motivated by phenomena in the physical world, the study of Ising models has identified phase transitions in the behavior of the model as its parameters vary. A common transition occurs as the temperature of the model changes from high to low. As the parameters $\vec \theta$ correspond to inverse (individualistic) temperatures, this corresponds to a transition of these parameters from low values (high temperature) to high values (low temperature). Often the transition to high temperature is identified with the satisfaction of Dobrushin-type conditions~\cite{Georgii11}.
Dobrushin's original theorem \cite{Dobruschin68} proposed a condition under which he proved uniqueness of an equilibrium Gibbs measure for infinite spin systems. Under his condition, finite spin Ising models also enjoy a number of good properties, including rapid mixing of the Glauber dynamics which in turn implies polynomial time sampleability, and spatial mixing properties (correlation decay). 
The Ising model has been studied extensively in such high-temperature regimes~\cite{Dobrushin56,Dobruschin68,Chatterjee05, Hayes06, DyerGJ08}, and it is a regime that is often used in practice.

In the high-temperature regime, we show that we can improve our baseline result without making ferromagnetic or tree-structure assumptions, using a non-localization based argument, explained next. 
In particular, we show in Theorem~\ref{thm:learn-and-test} that under high temperature and with no external fields independence testing can be done computationally efficiently from $\tilde{O}\left({n^{10/3}\beta^2 \over \epsilon^2 }\right)$ samples.
Since in the high-temperature regime, $\beta$ decreases as the degree increases, this is an improvement upon our baseline result in high-degree graphs.
Similar improvements hold when external fields are present (Theorem~\ref{thm:learn-then-test-ind-extfield-balanced}), as well as for identity testing, without and with external fields (Theorems~\ref{thm:learn-then-test-id-noextfield-balanced} and \ref{thm:learn-then-test-id-extfield-balanced}).

In the special case when the Ising model is both high-temperature \emph{and} ferromagnetic, we can use a similar algorithm to achieve a sample complexity of $\tilde O(n/\ve)$ (Theorem~\ref{thm:htferro-independence})\footnote{Prior work of~\cite{GheissariLP18} proves a qualitatively similar upper bound to ours, using a $\chi^2$-style statistic.
We show that our existing techniques suffice to give a near-optimal sample complexity.}.
This is nearly-tight for this case, as the lower bound instance of Theorem~\ref{thm:linear-lb} (which requires $\Omega(n/\ve)$ samples) is both high-temperature and ferromagnetic.
All of our results up to this point have been obtained via localization: blaming the distance of $p$ from the class of interest on a single node or edge.
Our improved bounds employ global statistics that look at all the nodes of the Ising model simultaneously. 
%
Specifically, we employ bilinear statistics of the form $Z = \sum_{e = (u,v) \in E} c_e X_u X_v$ for some appropriately chosen signs $c_e$.

The first challenge we encounter here involves selecting the signs $c_e$.
Ideally, we would like each $c_{uv}$ to match the sign of $\E[X_u X_v]$, which would result in a large gap in the expectation of $Z$ in the two cases.
While we could estimate these signs independently for each edge, this would incur an unnecessary overhead of $O(n^2)$ in the number of samples. 
Instead, we try to learn signs that have a non-trivial correlation with the correct signs from fewer samples. 
Despite the $X_uX_v$ terms potentially having nasty correlations with each other, a careful analysis using anti-concentration allows us to sidestep this $O(n^2)$ cost and generate satisfactory estimates from fewer samples.

The second and more significant challenge involves bounding the variance of a statistic $Z$ of the above form.
Since $Z$'s magnitude is at most $O(n^2)$, its variance can trivially be bounded by $O(n^4)$.
However, applying this bound in our algorithm gives a vacuous sample complexity.
As the $X_u$'s will experience a complex correlation structure, it is not clear how one might arrive at non-trivial bounds for the variance of such statistics, leading to the following natural question:
\begin{question}
How can one bound the variance of statistics over high-dimensional distributions? 
\end{question}
This meta-question is at the heart of many high-dimensional statistical tasks, and we believe it is important to develop general-purpose frameworks for such settings.
In the context of the Ising model, in fairly general regimes, we can show the variance to be $O(n^2)$.
Some may consider this surprising -- stated another way, despite the complex correlations which may be present in the Ising model, the summands in $Z$ behave roughly as if they were pairwise independent.
Our approach uses tools from~\cite{LevinPW09}.
It requires a bound on the spectral gap of the Markov chain, and an expected Lipschitz property of the statistic when a step is taken at stationarity.
The technique is described in Section~\ref{sec:variance}, and the variance bounds are given in Theorems~\ref{thm:variance bound high temperature no external field yuval} and~\ref{thm:variance bound high temperature with external field yuval}. 
\paragraph{Lower Bounds.}
We present information-theoretic lower bounds on the sample complexity of testing Ising models via Le Cam's method~\cite{LeCam73}. 
Inspired by Paninski's lower bound for univariate uniformity testing~\cite{Paninski08}, a first approach chooses a fixed matching of the nodes and randomly perturbs the weight of the edges, giving an $\Omega(\sqrt{n})$ lower bound on the sample complexity.
Though this will be superceded by our linear lower bound, it serves as a technical warm up,
while also proving a lower bound for uniformity testing on product distributions over a binary alphabet (which are a special case of the Ising model where no edges are present).

To achieve a linear lower bound, we instead consider a \emph{random} matching of the nodes.
The analysis of this case turns out to be much more involved due to the complex structure of the probability function which corresponds to drawing $k$ samples from an Ising model on a randomly chosen matching.
Indeed, our proof turns out to have a significantly combinatorial flavor, and we believe that our techniques might be helpful for proving stronger lower bounds in combinatorial settings for multivariate distributions.
Our analysis of this construction is tight, as uniformity testing on forests can be achieved with $O(n)$ samples.

Our lower bounds are presented in Section~\ref{sec:lb}.
See Theorem~\ref{thm:product-lb} for our $\Omega(\sqrt{n})$ lower bound, Theorem~\ref{thm:linear-lb} for our $\Omega(n)$ lower bound, and Theorem~\ref{thm:bh-lb} for our lower bounds which depend on $\b$ and $h$.

\bigskip
\noindent Table \ref{summary} summarizes our algorithmic results. 

\bgroup
\def\arraystretch{1.3}%
\begin{table*}[t] 
\centering
  \begin{tabular}{| c | c | c |}
  \hline
  {\bf Testing Problem} & {\bf No External Field} & {\bf Arbitrary External Field} \\ 
  \hline
    \begin{tabular}{@{}c@{}}\textsc{Independence} \\ using Localization\end{tabular} & $\tilde{O}\left(\frac{n^2\dm^2\b^2}{\ve^2} \right)$& $\tilde{O}\left(\frac{n^2\dm^2\b^2}{\ve^2} \right)$\\
  \hline
  \begin{tabular}{@{}c@{}}\textsc{Identity} \\ using Localization \end{tabular}
  & $\tilde{O}\left(\frac{n^2\dm^2\b^2}{\ve^2}\right)$ & $\tilde{O}\left(\frac{n^2\dm^2\b^2}{\ve^2} + \frac{n^2h^2}{\ve^2}\right)$ \\
  \hline
   \begin{tabular}{@{}c@{}}\textsc{Independence} \\ under Dobrushin/high-temperature \\ using Learn-Then-Test \end{tabular}
  & $\tilde{O}\left(\frac{n^{10/3}\b^2}{\ve^2} \right)$ & $\tilde{O}\left(\frac{n^{10/3}\b^2}{\ve^2} \right)$ \\
  \hline
   \begin{tabular}{@{}c@{}} \textsc{Identity} \\ under Dobrushin/high-temperature \\ using Learn-Then-Test \end{tabular}
  & $\tilde{O}\left(\frac{n^{10/3}\b^2}{\ve^2} \right)$ & $\tilde{O}\left(\frac{n^{11/3}\b^2}{\ve^2} + \frac{n^{5/3}h^2}{\ve^2}\right)$ \\
  \hline
   \begin{tabular}{@{}c@{}}\textsc{Independence on Forests} \\ using Improved Localization \end{tabular}
  & $\tilde{O}\left(\frac{n}{\ve} \right)$ & $\tilde{O}\left(\frac{n^2\b^2}{\ve^2} \right)$ \\
  \hline
   \begin{tabular}{@{}c@{}} \textsc{Identity on Forests} \\ using Improved Localization \end{tabular}
  & $\tilde{O}\left(\frac{n \cdot c(\b)}{\ve} \right)$ & $\tilde{O}\left(\frac{n^2\b^2}{\ve^2} +\frac{n^2h^2}{\ve^2}\right)$ \\
  \hline
   \begin{tabular}{@{}c@{}} \textsc{Independence on Ferromagnets} \\ using Improved Localization \end{tabular}
  & $\tilde{O}\left(\frac{n\dm}{\ve} \right)$ & $\tilde{O}\left(\frac{n^2\dm^2\b^2}{\ve^2} \right)$ \\
  \hline
     \begin{tabular}{@{}c@{}} \textsc{Independence on Ferromagnets} \\ under Dobrushin/high-temperature \\ using Globalization \end{tabular}
  & $\tilde{O}\left(\frac{n}{\ve} \right)$ & $\tilde{O}\left(\frac{n^{11/3}\b^2}{\ve^2} + \frac{n^{5/3}h^2}{\ve^2}\right)$ \\
  \hline
\end{tabular}
\caption{\label{summary} Summary of our results in terms of the sample complexity upper bounds for the various problems studied. $n=$ number of nodes in the graph, $\dm=$ maximum degree, $\b=$ maximum absolute value of edge parameters, $h=$ maximum absolute value of node parameters (when applicable), and $c$ is a function discussed in Theorem~\ref{thm:forests-identity}.}
\end{table*} 
\egroup

\subsection{Related Work}
The seminal works of Goldreich, Goldwasser, and Ron~\cite{GoldreichGR96}, Goldreich and Ron~\cite{GoldreichR00}, and Batu, Fortnow, Rubinfeld, Smith, and White~\cite{BatuFRSW00} initiated this study of distribution testing, viewing distributions as a natural domain for property testing (see~\cite{Goldreich17} for coverage of this much broader field).

Since these works, distribution testing has enjoyed a wealth of study, resulting in a thorough understanding of the complexity of testing many distributional properties (see e.g.~\cite{BatuFFKRW01, BatuKR04, Paninski08, AcharyaDJOP11, BhattacharyyaFRV11, Valiant11, Rubinfeld12, IndykLR12, DaskalakisDSVV13, LeviRR13, ChanDVV14, ValiantV17a, Waggoner15, AcharyaD15, AcharyaDK15, BhattacharyaV15, DiakonikolasKN15a, DiakonikolasK16, Canonne16, CanonneDGR16, BlaisCG17, BatuC17, DaskalakisKW18, DiakonikolasGPP18}, and~\cite{Rubinfeld12,Canonne15a, BalakrishnanW18, Kamath18} for recent surveys).
For many problems, these works have culminated in sample-optimal algorithms.

Motivated by lower bounds for testing properties of general discrete distributions, there has been a line of work studying distribution testing under structural restrictions.
In the univariate setting for example, one could assume the distribution is log-concave or $k$-modal.
This often allows exponential savings in the sample complexity~\cite{BatuKR04, DaskalakisDSVV13, DiakonikolasKN15a, DiakonikolasKN15b, DiakonikolasKN17, DiakonikolasKP19,CanonneKMUZ19}.

In the multivariate setting, such structural assumptions are even more important, as they allow one to bypass the curse of dimensionality incurred due to domain sizes growing exponentially in the dimension.
The present work fits into this paradigm by assuming the underlying distribution is an Ising model.
Concurrent works study distribution testing on Bayesian networks\footnote{Bayes nets are another type of graphical model, and are in general incomparable to Ising models.} \cite{CanonneDKS17, DaskalakisP17}, with more recent work on learning and testing \emph{causal} Bayesian networks with interventions~\cite{AcharyaBDK18}.
One may also consider testing problems in settings involving Markov Chains, of which there has been interest in testing standard properties as well as domain specific ones (i.e., the mixing time)~\cite{BatuFFKRW01, BhattacharyaV15, DaskalakisDG18, HsuKS15, LevinP16, HsuKLPS17,BerthetK19}.
Following the initial posting of our paper, there have also been other recent works on learning and testing Ising models, in both the statistical and structural sense~\cite{GangradeNS17, NeykovL17, DevroyeMR18a, BreslerN18, MukherjeeR19, BezakovaBCSV19}.
We highlight the last of these~\cite{BezakovaBCSV19}, which further explores the problem of identity testing.
The authors give an improved algorithm for testing ferromagnetic Ising models, and strong lower bounds for the case when one is giving weaker access to the reference distribution than we consider in this paper.
It remains to be seen which other multivariate distribution classes of interest allow us to bypass the curse of dimensionality.

Ising models, especially the problem of learning their structure (i.e., determining which edges are present in the graph), have enjoyed much study, especially in information theory~\cite{ChowL68, AbbeelKN06, CsiszarT06, RavikumarWL10, JalaliJR11, JalaliRVS11, SanthanamW12, BreslerGS14b, Bresler15, VuffrayMLC16, BreslerK16, BhattacharyaM16, MartindelCampoCU16, HamiltonKM17, KlivansM17, MukherjeeMY18, WuSD18, Bhattacharya19}.
At a first glance, one may hope that results on structure learning of Ising models have implications for testing.
However, thematic similarities aside, the two problems are qualitatively very different -- our problem focuses on statistical estimation, while theirs looks at structural estimation.
While we do not include a full comparison between the problems in this paper, we point out some qualitative differences: the complexity of structure learning is exponential in the maximum degree and $\b$, while only logarithmic in $n$.
On the other hand, for testing Ising models, the complexity has a polynomial dependence in all three parameters, which is both necessary and sufficient.


Understanding the behavior of statistics in multivariate settings (which may have complex correlations between random variables) is crucial for data analysis in high-dimensional settings.
Our approach requires us to bound the variance of bilinear statistics of the Ising model in the high-temperature regime.
As we demonstrate, despite the existence of potentially rich correlation structures, this quantity is (up to log factors) the same as if the random variables were all independent.
Following the original posting of this paper, concurrent works by the authors~\cite{DaskalakisDK17} and Gheissari, Lubetzky, and Peres~\cite{GheissariLP18} investigated additional properties of multilinear functions of the Ising model.
In particular, these works prove \emph{concentration} bounds (which are qualitatively stronger than variance bounds) for multilinear functions of arbitrary degree $d$ (rather than just bilinear functions, which are of degree $d=2$).
Further works sharpen these results both quantitatively (by narrowing the radius of concentration and weight of the tails) and qualitatively (with multilevel concentration results)~\cite{GotzeSS18,AdamczakKPS19}.
Even more recently, a line of study has investigated supervised learning problems under limited dependency between the samples (rather than the usual i.i.d.\ setting), using related techniques~\cite{DaskalakisDP19, DaganDDJ19}.

The subsequent work by Gheissari, Lubetzky, and Peres~\cite{GheissariLP18} applies their higher-order variance and concentration bounds to give an improved statistic for testing independence of high-temperature Ising models, including results for when the distribution is further known to be ferromagnetic.
In particular, their method uses a $\chi^2$-style statistic, giving rise to fourth-order multilinear terms.
Their results allow them to control the variance of this higher-order statistic.
Following their result, we show that the same result for high-temperature ferromagnets can be achieved by our bilinear statistic in Section~\ref{sec:htferro}.

\section{Preliminaries}
\label{sec:preliminaries}
Recall the definition of the Ising model from Eq.~(\ref{eq:ising model}). We will abuse notation, referring to both the probability distribution $p$ and the random vector $X$ that it samples in $\{\pm 1\}^V$ as the Ising model. That is, $X \sim p$. We will use $X_u$ to denote the variable corresponding to node $u$ in the Ising model $X$. When considering multiple samples from an Ising model $X$, we will use $X^{(l)}$ to denote the $l^{th}$ sample. We will use $h$ to denote the largest node parameter in absolute value and $\b$ to denote the largest edge parameter in absolute value. That is, $|\th_v| \leq h$ for all $v \in V$ and $|\th_e| \leq \b$ for all $e \in E$. Depending on the setting, our results will depend on $h$ and $\b$.
Furthermore, in this paper we will use the convention that $E = \{(u,v)\ |\ u,v \in V, u \neq v\}$ and $\theta_e$ may be equal to $0$, indicating that edge $e$ is not present in the graph.
We use $m$ to denote the number of edges with non-zero parameters in the graph, and $\dm$ to denote the maximum degree of a node.
Throughout this paper, we will use the notation $\m_v \triangleq \E[X_v]$ for the marginal expectation of a node $v \in V$ (also called the node marginal), and similarly $\m_{uv} \triangleq \E[X_u X_v]$ for the marginal expectation of an edge $e = (u,v) \in E$ (also called the edge marginal). 
In case a context includes multiple Ising models, we will use $\m_e^p$ to refer to the marginal expectation of an edge $e$ under the model $p$.
We will use $\U$ to denote the uniform distribution over $\{\pm 1\}^n$, which also corresponds to the Ising model with $\vec \th = \vec 0$.
Similarly, we use $\I$ for the set of all product distributions over $\{\pm 1\}^n$. 
In this paper, we will consider \emph{Rademacher} random variables, where $Rademacher(p)$ takes value $1$ with probability $p$, and $-1$ otherwise.
When $\vec{p}$ and $\vec{q}$ are vectors, we will write $\vec{p} \leq \vec{q}$ to mean that $p_i \leq q_i$ for all $i$.

\begin{definition}
In the setting with \emph{no external field}, $\th_v = 0$ for all $v \in V$.
\end{definition}

\begin{definition}
In the \emph{ferromagnetic} setting, $\th_e \geq 0$ for all $e \in E$.
\end{definition}

\begin{definition}[Dobrushin's Uniqueness Condition]
\label{def:dobrushin-new}
Consider an Ising model $p$ defined on a graph $G=(V,E)$ with $|V|=n$ and parameter vector $\vec{\th}$.
Suppose $\max_{v \in V} \sum_{u \ne v} \tanh\left(\abss{\th_{uv}}\right) \le  1-\eta$ for some constant $\eta > 0$. Then $p$ is said to satisfy Dobrushin's uniqueness condition, or be in the high temperature regime. 
Note that since $\tanh(|x|) \le |x|$ for all $x$, the above condition follows from more simplified conditions which avoid having to deal with hyperbolic functions. 
For instance, either of the conditions 
$\max_{v \in V} \sum_{u \ne v} \abss{\th_{uv}} \le  1-\eta$  or $
\b \dm \le  1-\eta$
are sufficient to imply Dobrushin's condition (where $\b = \max_{u,v} \abss{\th_{uv}}$ and $\dm$ is the maximum degree of $G$).
\end{definition}
In general, when one refers to the temperature of an Ising model, a high temperature corresponds to small $\theta_e$ values, and a low temperature corresponds to large $\theta_e$ values.
In this paper, we will only use the precise definition as given in Definition~\ref{def:dobrushin-new}.

In Section~\ref{sec:variance} we use the Glauber dynamics on the Ising model. Glauber dynamics is the canonical Markov chain for sampling from an Ising model. The dynamics define a reversible, ergodic chain whose stationary distribution is the corresponding Ising model. We consider the basic variant known as single-site Glauber dynamics. The dynamics are defined on the set $\Sigma^n$ where $\Sigma = \{\pm 1\}$.  We also refer to this set as $\Omega$ in our work. They proceed as follows:
\begin{enumerate}
	\item Start at any state $X^{(0)} \in \Sigma^n$. Let $X^{(t)}$ denote the state of the dynamics at time $t$.
	\item Let $N(u)$ denote the set of neighbors of node $u$. Pick a node $u$ uniformly at random and update $X_u$ as follows:
	\begin{align*}
		&~~ X_u^{(t+1)} = x \quad \text{w.p. } \quad \frac{\exp\left(\th_u x + \sum_{v \in N(u)} \th_{uv}X_v^{(t)}x \right)}{\exp\left(\th_u + \sum_{v \in N(u)} \th_{uv}X_v^{(t)} \right) + \exp\left(-\th_u-\sum_{v \in N(u)} \th_{uv}X_v^{(t)} \right)} 
	\end{align*}
\end{enumerate}

\begin{remark}
\label{rm:contraction}
We note that high-temperature is not strictly needed for our results to hold -- we only need Hamming contraction of the ``greedy coupling.''
This condition implies rapid mixing of the Glauber dynamics (in $O(n \log n)$ steps) via path coupling (Theorem 15.1 of~\cite{LevinPW09}).
See~\cite{DaskalakisDK17,GheissariLP18} for further discussion of this weaker condition.
\end{remark}

Functions of the Ising model satisfying a generalized Lipschitz condition (bounded differences) have the following variance bound, which is in Chatterjee's thesis~\cite{Chatterjee05}:

\begin{lemma}[Lipschitz Concentration Lemma]
	\label{lem:lipschitz-lemma}
	Suppose that $f(X_1,\ldots,X_n)$ is a function of an Ising model in the high-temperature regime. 
	Suppose the bounded differences constants of $f$ are $l_1, l_2, \ldots, l_n$ respectively. That is,
	$\abss{f(X_1,\ldots,X_i,\ldots,X_n) - f(X_1,\ldots,X_i',\ldots,X_n)} \le l_i$ 
	for all values of $X_1,\allowbreak \ldots,\allowbreak X_{i-1},\allowbreak X_{i+1},\allowbreak \ldots,\allowbreak X_n$ and for any $X_i$ and $X_i'$. Then for some absolute constant $c_1$,
	$$\Pr\left[ \abss{f(X)-\E[f(X)]} > t\right] \le 2\exp\left(-\frac{c_1 t^2}{2\sum_{i=1}^n l_i^2} \right).$$
	In particular, for some absolute constant $c_2$,
	$\Var(f(X)) \leq c_2 \sum_{i} l_i^2$.
\end{lemma}

We will use the symmetric KL divergence, defined as follows:
$$\dskl(p,q) = \dkl(p,q) + \dkl(q,p) = \E_p\left[\log \left(\frac{p}{q}\right) \right] + \E_q\left[ \log \left(\frac{q}{p}\right)\right].$$
Recall that this upper bounds the vanilla KL divergence, and by Pinsker's inequality, also upper bounds the square of the total variation distance.

The symmetric KL divergence between two Ising models $p$ and $q$ admits a very convenient expression \cite{SanthanamW12}:
\begin{eqnarray}
\label{eq:dskl}
\dskl(p, q) = \sum_{v \in V} \left(\theta^{p}_v - \theta^{q}_v\right)\left(\m_v^{p} - \m^{q}_v\right)
 + \sum_{e = (u,v) \in E} \left(\theta^{p}_e - \theta^{q}_e\right)\left(\m^{p}_e - \m^{q}_e\right).
\end{eqnarray}

\paragraph{Input to Goodness-of-Fit Testing Algorithms.}
To solve the goodness-of-fit testing or identity testing problem with respect to a discrete distribution $q$, a description of $q$ is given as part of the input along with sample access to the distribution $p$ which we are testing. 
In case $q$ is an Ising model, its support has exponential size and specifying the vector of probability values at each point in its support is inefficient. 
Since $q$ is characterized by the edge parameters between every pair of nodes and the node parameters associated with the nodes, a succinct description would be to specify the parameters vectors $\{\th_{uv}\},\{\th_u\}$. 
In many cases, we are also interested in knowing the edge and node marginals of the model. 
Although these quantities can be computed from the parameter vectors, there is no efficient method known to compute the marginals exactly for general regimes~\cite{BreslerGS14a,Montanari15}.
A common approach is to use MCMC sampling to generate samples from the Ising model. 
However, for this technique to be efficient we require that the mixing time of the Markov chain be small which is not true in general. 
Estimating and exact computation of the marginals of an Ising model is a well-studied problem but is not the focus of this paper. 
Hence, to avoid such computational complications we will assume that for the identity testing problem the description of the Ising model $q$ includes both the parameter vectors $\{\th_{uv}\},\{\th_u\}$ as well as the edge and node marginal vectors $\{\m_{uv}=\E[X_uX_v] \}, \{\m_u = \E[X_u] \}$.

\section{A General Purpose Localization Algorithm}
\label{sec:localization}
Our first algorithm is a general purpose ``localization'' algorithm.
While simple, this serves as a proof-of-concept that testing on Ising models can avoid the curse of dimensionality, while simultaneously giving a very efficient algorithm for certain parameter regimes. 
The main observation which enables us to do a localization based approach is stated in the following lemma, which allows us to ``blame'' a difference between models $p$ and $q$ on a discrepant node or edge.
\begin{lemma}
\label{lem:localization-skl}
Given two Ising models $p$ and $q$, if $\dskl(p,q) \geq \ve$, then either
\begin{itemize}
\item There exists an edge $e=(u,v)$ such that $\left| \m_{uv}^p-\m_{uv}^q\right| \geq \frac{\ve}{4m\beta}$; or
\item There exists a node $u$ such that $\left| \m_{u}^p-\m_{u}^q\right| \geq  \frac{\ve}{4nh}$.
\end{itemize}
\end{lemma}
\begin{proof}
(\ref{eq:dskl}) implies $ \sum_{e = (u,v) \in E} \left(\theta^{p}_e - \theta^{q}_e\right)\left(\m^{p}_e - \m^{q}_e\right) \geq \ve/2$ or $\sum_{v \in V} \left(\theta^{p}_v - \theta^{q}_v\right)\left(\m_v^{p} - \m^{q}_v\right) \geq \ve/2$.
In the first case, $\exists\ e=(u,v)$ such that $\left(\th_{uv}^p - \th_{uv}^q \right) \left( \m_{uv}^p-\m_{uv}^q\right) \geq \frac{\ve}{2m}$, and the result follows since $|\theta^p_{uv} - \theta^q_{uv}| \leq 2\beta$. 
The second case is similar.
\end{proof}

Before giving a description of the localization algorithm, we state its guarantees.

\begin{theorem}
\label{thm:localization}
Given $\tilde O\left(\frac{m^2 \b^2}{\ve^2}\right)$ samples from an Ising model $p$, there exists a polynomial-time algorithm which distinguishes between the cases $p \in \I$ and $\dskl(p, \I) \geq \ve$ with probability at least $2/3$. 
Furthermore, given $\tilde O\left(\frac{m^2 \b^2}{\ve^2} + \frac{n^2 h^2}{\ve^2}\right)$ samples from an Ising model $p$ and a description of an Ising model $q$, there exists a polynomial-time algorithm which distinguishes between the cases $p = q$ and $\dskl(p,q) \geq \ve$ with probability at least $2/3$. 
\end{theorem}
In this theorem, $m$ is an upper bound on the number of edges: if unknown, it can be taken as the trivial bound $\binom{n}{2}$, or $n \dm$ if a bound is given on the maximum degree $\dm$.
We describe the algorithm for independence testing in Section~\ref{subsec:localization-ind}. 
The algorithm for testing identity is similar, its description and correctness proofs are given in Section~\ref{subsec:localization-id}. 

\subsection{Independence Test using Localization}
\label{subsec:localization-ind}
Algorithm~\ref{alg:localization-product-test} computes the empirical estimate for the covariance of each pair of variables in the Ising model: $\lambda_{uv} = \m_{uv} - \m_u \m_v$  for all pairs $(u,v)$. 
If they are all close to zero, then we can conclude that $p \in \I$. 
If there exists an edge for which $\lambda_{uv}$ is far from 0, this indicates that $p$ is far from $\I$. 

\begin{algorithm}[htb]
\begin{algorithmic}[1]
\Function{LocalizationTest}{sample access to Ising model $p$, accuracy parameter $\ve, \b, m$}
\State \parbox[t]{\dimexpr\linewidth-\algorithmicindent}{Draw $k = O\left(\frac{m^2\b^2\log n}{\ve^2}\right)$ samples from $p$. Denote the samples by $X^{(1)},\ldots,X^{(k)}$.\strut}
\State \parbox[t]{\dimexpr\linewidth-\algorithmicindent}{Compute empirical estimates $\hat{\m}_u = \frac{1}{k}\sum_{i} X_u^{(i)}$ for each node $u \in V$ and $\hat{\m}_{uv} = \frac{1}{k}\sum_{i} X_u^{(i)}X_v^{(i)}$ for each pair of nodes $(u,v)$.\strut}
\State \parbox[t]{\dimexpr\linewidth-\algorithmicindent}{Using the above estimates compute the covariance estimates $\hat{\lambda}_{uv} = \hat{\m}_{uv}-\hat{\m}_u\hat{\m}_v$ for each pair of nodes $(u,v)$.\strut}
\State \parbox[t]{\dimexpr\linewidth-\algorithmicindent}{If for any pair of nodes $(u,v)$, $\abss{\hat{\lambda}_{uv}} \geq \frac{\ve}{4m\b}$ return that $\dskl(p, \I) \geq \ve$.\strut}
\State Otherwise, return that $p \in \I$.
\EndFunction
\end{algorithmic}
\caption{Test if an Ising model $p$ is product}
\label{alg:localization-product-test}
\end{algorithm}

To prove correctness of Algorithm \ref{alg:localization-product-test}, we will require the following lemma, which allows us to detect pairs $u,v$ for which $\lambda_{uv}$ is far from 0.
\begin{lemma}
\label{lem:localization}
Given $O\left(\frac{\log n }{\ve^2} \right)$ samples from an Ising model $X \sim p$, there exists an algorithm which, with probability at least $9/10$, can identify all pairs of nodes $(u,v) \in V^2$ such that $|\lambda_{uv}| \geq \ve$, where $\lambda_{uv} = \E[X_uX_v]-\E[X_u]\E[X_v]$.
Namely, the algorithm computes the empirical value of $|\lambda_{uv}|$ for each pair of nodes and identifies pairs such that this value is sufficiently far from $0$.
\end{lemma}
\begin{proof}
For an edge $e = (u,v)$, with $O(\log n/ \ve^2)$ samples, a Chernoff bound implies we can obtain estimates such that $|\hat{\m}_{uv} - \m_{uv}|, |\hat{\m}_u - \mu_u|, |\hat{\m}_v - \mu_v|$ are all $\leq \frac{\ve}{10}$ with probability at least $1 - 1/10n^2$, which we condition on.
Let $\hat{\lambda}_{uv} = \hat{\m}_{uv}-\hat{\m}_u\hat{\m}_v$. 
By triangle inequality, $|\lambda_{uv}-\hat{\lambda}_{uv}| \leq \frac{3\ve}{10} + \frac{\ve^2}{100} \leq \frac{2\ve}{5}$. 
If $\abss{\l_{uv}} \geq \ve$, then $\abss{\hat{\lambda}_{uv}} \geq \frac{3\ve}{5}$, and if $\abss{\l_{uv}} = 0$, then $\abss{\hat{\l}_{uv}} \leq \frac{2\ve}{5}$, so thresholding $\abss{\hat{\l}_{uv}}$ at $\frac{\ve}{2}$ lets us detect whether $\lambda_{uv}=0$ or $|\lambda_{uv}| \geq \ve$ with probability at least $1 - 1/10n^2$.
The result follows by a union bound.
\end{proof}

With this lemma in hand, we now prove the first part of Theorem~\ref{thm:localization}.
\begin{lemma}
\label{lem:localization-product}
Given $\tilde O\left(\frac{m^2 \b^2}{\ve^2}\right)$ samples from an Ising model $X \sim p$, Algorithm~\ref{alg:localization-product-test} distinguishes between the cases $p \in \I$ and $\dskl(p, \I) \geq \ve$ with probability at least $2/3$.
\end{lemma}
\begin{proof}
If $p \in \I$, then $\mu^p_{uv} = \mu^p_v \mu^p_v$ and thus $\lambda_{uv} = 0$ for all edges $e = (u,v)$.
On the other hand, if $\dskl(p, \I) \geq \ve$, then in particular $p$ will be far from the product distribution $q$ on $n$ nodes such that $\m_u^q = \m_u^p$ for all $u \in V$. 
By minor modifications to Lemma~\ref{lem:localization-skl} (since $\mu_u^p = \mu_u^q$ for all nodes $u$, and $\theta^q_e = 0$ for all edges $e$), there will exist $e = (u,v)$ such that $|\mu_{uv}^p - \mu_{uv}^q| = |\lambda_{uv}| \geq \frac{\ve}{2m\beta}$.
Note that Algorithm~\ref{alg:localization-product-test} simply runs the test of Lemma~\ref{lem:localization}, so the result follows by its guarantees.
\end{proof}

\subsection{Identity Test using Localization}
\label{subsec:localization-id}
Along the same lines, one can also test identity to an Ising model $q$.
In fact, this time the approach is even simpler, as we simply compare the empirical node and edge marginals from $p$ with the node and edge marginals of $q$.
As before, we can quantify the empirical estimation of these parameters using a Chernoff bound.
Since the argument is very similar to before, we omit the details, and state the guarantees of Algorithm~\ref{alg:localization-identity-test}, which allows us to conclude Theorem~\ref{thm:localization}.
\begin{lemma}
\label{lem:localization-identity}
Given $\tilde O\left(\frac{m^2 \b^2}{\ve^2} + \frac{n^2 h^2}{\ve^2}\right)$ samples from an Ising model $X \sim p$ and a description of an Ising model $q$, Algorithm~\ref{alg:localization-identity-test} distinguishes between the cases $p=q$ and $\dskl(p, q) \geq \ve$ with probability at least $2/3$.
\end{lemma}

\begin{algorithm}[htb]
\begin{algorithmic}[1]
\Function{LocalizationTestIdentity}{sample access to Ising model $X \sim p$, description of Ising model $q$, accuracy parameter $\ve$,$\b$,$h$,$m$}
\State \parbox[t]{\dimexpr\linewidth-\algorithmicindent}{Draw $k = O\left(\frac{\left(m^2\b^2+n^2h^2\right)\log n}{\ve^2}\right)$ samples from $p$. Denote the samples by $X^{(1)},\ldots,X^{(k)}$.\strut}
\State \parbox[t]{\dimexpr\linewidth-\algorithmicindent}{Compute empirical estimates $\hat{\m}_u^p = \frac{1}{k}\sum_{i} X_u^{(i)}$ for each node $u \in V$ and $\hat{\m}_{uv}^p = \frac{1}{k}\sum_{i} X_u^{(i)}X_v^{(i)}$ for each pair of nodes $(u,v)$.\strut}
\State \parbox[t]{\dimexpr\linewidth-\algorithmicindent}{If for any pair of nodes $(u,v)$, $\abss{\hat{\m}_{uv}^p-\m_{uv}^q} \geq \frac{\ve}{8m\b}$ return that $\dskl(p, q) \geq \ve$.\strut}
\State \parbox[t]{\dimexpr\linewidth-\algorithmicindent}{If for any node $u$, if $\abss{\hat{\m}_{u}^p-\m_u^q} \geq \frac{\ve}{8n h}$ return that $\dskl(p,q) \geq \ve$.\strut}
\State Otherwise, return that $p = q$.
\EndFunction
\end{algorithmic}
\caption{Test if an Ising model $p$ is identical to $q$}
\label{alg:localization-identity-test}
\end{algorithm}


\section{Improved Tests for Forests and Ferromagnetic Ising Models}
\label{sec:stronglocalization}

In this section we will describe testing algorithms for two commonly studied classes of Ising models, namely forests and ferromagnets.
In these cases, the sample complexity improves compared to the baseline result when in the regime of no external field. 
The testers are still localization based (like those of Section~\ref{sec:localization}), but we can now leverage structural properties to obtain more efficient testers.

First, we consider the class of all forest structured Ising models, where the underlying graph $G=(V,E)$ is a forest. 
Under no external field, it can be shown that the edge marginals $\m_e$ have a simple closed form expression. 
This structural information enables us to improve our testing algorithms from Section \ref{sec:localization} on forest graphs. 
We state the improved sample complexities here and defer a detailed description of the algorithms to Section~\ref{sec:forests}. 

\begin{theorem}[Independence testing of Forest-Structured Ising Models]
\label{thm:forests-independence}
Algorithm \ref{alg:forests-independence} takes in $\tilde{O}\left( \frac{n}{\ve}\right)$ samples from an Ising model $X \sim p$ whose underlying graph is a forest and which is under no external field and outputs whether $p \in \I$ or $\dskl(p,\I) \geq \ve$ with probability $\geq 9/10$.
\end{theorem}
\begin{remark}
Note that Theorem~\ref{thm:forests-independence} together with our lower bound described in Theorem~\ref{thm:linear-lb} indicate a tight sample complexity up to logarithmic factors for independence testing on forest-structured Ising models under no external field.
\end{remark}

\begin{theorem}[Identity Testing of Forest-Structured Ising Models]
\label{thm:forests-identity}
Algorithm \ref{alg:forests-identity} takes in the edge parameters of an Ising model $q$ on a forest graph and under no external field as input, and draws $\tilde{O}\left(c(\b) \frac{n}{\ve}\right)$ samples from an Ising model $X \sim p$ (where $c(\b)$ is a function of the parameter $\b$) whose underlying graph is a forest and under no external field, and outputs whether $p = q$ or $\dskl(p,q) \geq \ve$ with probability $\geq 9/10$.
\end{theorem}
Note that for identity testing, any algorithm necessarily has to have at least a $\b$ dependence due to the lower bound we show in Theorem \ref{thm:bh-lb}.\\

The second class of Ising models we consider this section are ferromagnets. 
For a ferromagnetic Ising model, $\th_{uv} \geq 0$ for every pair of nodes $u,v$. 
Ferromagnets may potentially contain cycles but since all interactions are ferromagnetic, the marginal of every edge is at least what it would have been if it was a lone edge. 
We prove this with an application of Griffiths inequality, and using this structural property, we give a quadratic improvement in the dependence on parameter $m$ for testing independence under no external field. 
We state our main result in this regime here and a full description of the algorithm and the structural lemma are provided in Section~\ref{sec:ferro}.

\begin{theorem}[Independence Testing of Ferromagnetic Ising Models]
\label{thm:ferro-independence}
Algorithm \ref{alg:ferro} takes in $\tilde{O}\left( \frac{n\dm}{\ve}\right)$ samples from a ferromagnetic Ising model $X \sim p$ which is under no external field and outputs whether $p \in \I$ or $\dskl(p,\I) \geq \ve$ with probability $\geq 9/10$.
\end{theorem}

\subsection{Improved Algorithms for Independence and Identity Testing on Forests}
\label{sec:forests}

We use the following folklore fact on the edge marginals of forest-structured Ising models; for completeness, we provide a proof in Appendix~\ref{sec:treemarginalAppendix}.
\begin{lemma}[Structural Lemma for Forest-Structured Ising Models] \label{lem:trees-structural}
If $p$ is an Ising model on a forest graph with no external field, and $X \sim p$, then for any $(u,v) \in E$, 
$\E\left[ X_u X_v\right] = \tanh(\th_{uv})$.
\end{lemma}

Given the above structural lemma, we give the following localization-style algorithm for testing independence on forest Ising models under no external field.

\begin{algorithm}[htb]
\begin{algorithmic}[1]
\Function{TestForestIsing-Product}{sample access to Ising model $p$}
\State \parbox[t]{\dimexpr\linewidth-\algorithmicindent}{Run the algorithm of Lemma \ref{lem:localization} to identify all edges $e=(u,v)$ such that $\abss{\E[X_uX_v]} \geq \sqrt{\frac{\ve}{n}}$ using $\tilde{O}\left(\frac{n}{\ve}\right)$ samples.
If it identifies any edges, return that $\dskl(p, \I) \geq \ve$.\strut}
\State Otherwise, return that $p$ is product.
\EndFunction
\end{algorithmic}
\caption{Test if a forest Ising model $p$ under no external field is product}
\label{alg:forests-independence}
\end{algorithm}


\begin{prevproof}{Theorem}{thm:forests-independence}
Firstly, note that under no external field, the only product Ising model is the uniform distribution $\U$. Therefore the problem reduces to testing whether $p$ is uniform or not. Consider the case when $p$ is indeed uniform. That is, there are no edges in the underlying graph of the Ising model. In this case with probability at least $9/10$ the localization algorithm of Lemma \ref{lem:localization} will output no edges. Hence Algorithm \ref{alg:forests-independence} will output that $p$ is uniform.\\
In case $\dskl(p,\U) \geq \ve$, we split the analysis into two cases.
\begin{itemize}
\item \textit{Case 1:} There exists an edge $e=(u,v)$ such that $\abss{\th_{uv}} \geq \sqrt{\frac{\ve}{n}}$. In this case, $\E[X_uX_v] = \tanh(\th_{uv})$ and in the regime where $\abss{\th} = o(1)$, $\abss{\tanh(\th)} \geq \abss{\th/2}$. Hence implying that $\abss{\E[X_uX_v]} \geq \abss{\th_{uv}/2} \geq \abss{\sqrt{\frac{\ve}{n}}/2}$. Therefore the localization algorithm of Lemma \ref{lem:localization} would identify such an edge with probability at least $9/10$. Note that the regime where the inequality $\abss{\tanh(\th)} \geq \abss{\th/2}$ isn't valid is easily detectable using $\tilde{O}(\frac{n}{\ve})$ samples, as this would imply that $|\theta| \geq 1.9$ and $\abss{\E[X_uX_v]} \geq 0.95$.
\item \textit{Case 2:} All edges $e=(u,v)$ are such that $\abss{\th_{uv}} \leq \abss{\sqrt{\frac{\ve}{n}}}$. 
Manipulation of (\ref{eq:dskl}) tells us that there exists an edge $(u,v)$ such that $\th_{uv}\E[X_uX_v] \geq \ve/n$.
The upper bound on $\th_{uv}$ indicates that $\E[X_uX_v] \geq \sqrt{\ve/n}$, and thus
 the localization algorithm of Lemma \ref{lem:localization} identifies this edge with probability at least $9/10$.

\end{itemize}
\end{prevproof}

Next, we will present an algorithm for identity testing on forest Ising models under no external field. 

\begin{algorithm}[htb]
\begin{algorithmic}[1]
\Function{TestForestIsing-Identity}{Ising model $q$,sample access to Ising model $p$}
\State \parbox[t]{\dimexpr\linewidth-\algorithmicindent}{If the Ising model $q$ is not a forest, or has a non-zero external field on some node, return $\dskl(p, q) \geq \ve$.\strut}
\State \parbox[t]{\dimexpr\linewidth-\algorithmicindent}{Run the algorithm of Lemma \ref{lem:localization} to identify all edges $e=(u,v)$ such that $\abss{\E[X_uX_v]-\tanh(\th_{uv}^q)} \geq \sqrt{\frac{\ve}{n}}$ using $\tilde{O}\left(\frac{n}{\ve}\right)$ samples.
If it identifies any edges, return that $\dskl(p, q) \geq \ve$.\strut}
\State Otherwise, return that $p=q$.
\EndFunction
\end{algorithmic}
\caption{Test if a forest Ising model $p$ under no external field is identical to a given Ising model $q$}
\label{alg:forests-identity}
\end{algorithm}

\begin{prevproof}{Theorem}{thm:forests-identity}
Consider the case when $p$ is indeed $q$. In this case with probability at least $9/10$ the localization algorithm of Lemma \ref{lem:localization} will output no edges. Hence Algorithm \ref{alg:forests-identity} will output that $p$ is uniform.\\
In case $\dskl(p,q) \geq \ve$, we split the analysis into two cases.
\begin{itemize}
\item \textit{Case 1:} There exists an edge $e=(u,v)$ such that $\abss{\th_{uv}^p-\th_{uv}^q} \geq \sqrt{\frac{\ve}{n}}$. In this case, $\E[X_uX_v] - \m_{uv}^q = \tanh(\th_{uv}^p) - \tanh(\th_{uv}^q)$ and hence has the same sign as $\th_{uv}^p-\th_{uv}^q$. Assume that $\th_{uv}^p \geq \th_{uv}^q$. The argument for the case $\th_{uv}^q > \th_{uv}^p$ will follow similarly. If $\th_{uv}^p-\th_{uv}^q \leq 1/2\tanh(\b)$, then the following inequality holds from Taylor's theorem.
$$\tanh(\th_{uv}^p)-\tanh(\th_{uv}^q) \geq \frac{\sech^2(\b)\left(\th_{uv}^p-\th_{uv}^q \right)}{2}$$
which would imply $\tanh(\th_{uv}^p)-\tanh(\th_{uv}^q) \geq \frac{\sech^2(\b)}{2}\sqrt{\frac{\ve}{n}}$ and hence the localization algorithm of Lemma~\ref{lem:localization} would identify edge $e$ with probability at least $9/10$ using $\tilde{O}\left(\frac{c_1(\b)n}{\ve}\right)$ samples (where $c_1(\b) = \cosh^4(\b)$).
If $\th_{uv}^p-\th_{uv}^q > 1/2\tanh(\b)$, then $\tanh(\th_{uv}^p)-\tanh(\th_{uv}^q) \geq \tanh(\b)-\tanh\left(\b-\frac{1}{2\tanh(\b)}\right)$ and hence the localization algorithm of Lemma~\ref{lem:localization} would identify edge $e$ with probability at least $9/10$ using $\tilde{O}\left( c_2(\b) \right)$ samples where $c_2(\b) = \frac{1}{(\tanh(\b)-\tanh(\b-1/2\tanh(\b)))^2}$. Note that as $\b$ grows small, $c_2(\b)$ gets worse. However it cannot grow unbounded as we also have to satisfy the constraint that $\th_{uv}^p - \th_{uv}^q \leq 2\b$. This implies that 
\[
c_2(\b) = \min\left\{\b^2, \frac{1}{(\tanh(\b)-\tanh(\b-1/2\tanh(\b)))^2} \right\}
\] samples suffice in this case. Therefore the algorithm will give the correct output with probability $> 9/10$ using $\tilde{O}\left(c(\b)\frac{n}{\ve} \right)$ samples where $c(\b) = \max\{c_1(\b),c_2(\b)\}$.

\item \textit{Case 2:} All edges $e=(u,v)$ are such that $\abss{\th_{uv}^q-\th_{uv}^q} \leq \sqrt{\frac{\ve}{n}}$. 
This case is similar to the corresponding case in the proof of Theorem~\ref{thm:forests-independence}.

\end{itemize}
\end{prevproof}

\subsection{Improved Algorithms for Independence Testing on Ferromagnets}
\label{sec:ferro}

We first prove a structural lemma about ferromagnetic Ising models. We will employ Griffiths inequality~\cite{Griffiths69} to argue that in any ferromagnetic Ising model $\mu_{uv} \geq \tanh(\th_{uv})$ for all pairs $u,v$. 

\begin{lemma}[Griffiths Inequality]
	\label{lem:griffiths}
	Let $p$ be a ferromagnetic Ising model under no external field defined on $G = (V,E)$. For any $A \subseteq V$ define $X_A = \prod_{u \in A}  X_u$ and let $\mu_{A}^p = \E\left[ X_A \right]$. For every $A,B \subseteq V$, 
	\begin{align*}
	\E\left[X_A X_B\right] \ge \E\left[ X_A\right]\E\left[X_B\right].
	\end{align*}
\end{lemma}

\begin{lemma} \label{lem:ferro-percolation}
Consider two ferromagnetic Ising models $p$ and $q$ under no external field defined on $G_p = (V,E_p)$ and $G_q = (V,E_q)$. Denote the parameter vector of $p$ model by $\vec{\th}^p$ and that of $q$ model by $\vec{\th}^q$. If $\vec{\th}^p \geq \vec{\th}^q$ coordinate-wise, then for any two nodes $u,v \in V$, $\mu_{uv}^p \geq \mu_{uv}^q$.
\end{lemma}
\begin{proof}
	
	We will show that for any ferromagnetic Ising model $p$ defined on $G = (V,E)$ under no external field, for all $i,j \in V$,
	$\frac{\partial \mu_{uv}^p}{\partial \th_{ij}^p} \ge 0$.
	It is clear that this will imply the lemma statement. 
	We drop the superscript $p$ from the following calculations for brevity.
	\begin{align}
	\frac{\partial \mu_{uv}}{\partial \th_{ij}} &= \frac{\partial}{\partial \th_{ij}} \frac{\sum_{x \in \Omega} x_u x_v \exp\left(\sum_{u \ne v} \th_{uv}x_u x_v\right)}{Z_{\vec{\th}}} \notag\\
	&= \frac{ \frac{\partial \left( \sum_{x \in \Omega} x_u x_v \exp\left(\sum_{u \ne v} \th_{uv}x_u x_v\right)  \right)}{\partial \th_{ij}}}{Z_{\vec{\th}}} - \frac{\mu_{uv} \frac{\partial \left(\sum_{x \in \Omega} \exp\left(\sum_{u \ne v} \th_{uv} x_u x_v\right) \right)}{\partial \th_{ij}}}{Z_{\vec{\th}}} \label{eq:fer-struct2}\\
	&= \E\left[X_uX_vX_iX_j\right] - \E\left[X_uX_v\right]\E\left[X_iX_j\right] \ge 0, \label{eq:fer-struct3}
	\end{align}
	where (\ref{eq:fer-struct2}) follows from the quotient rule for differentiation and (\ref{eq:fer-struct3}) follows by simple observation and finally to conclude non-negativity we used the Griffiths inequality (Lemma \ref{lem:griffiths}).
\end{proof}

Using the above lemma, we now prove the main structural lemma for ferromagnets which will be crucial to our algorithm for testing ferromagnetic Ising models.
\begin{lemma}[Structural Lemma about Ferromagnetic Ising Models]
\label{lem:ferro-structural}
If $X \sim p$ is a ferromagnetic Ising model on a graph $G=(V,E)$ under zero external field, then $\mu_{uv} \geq \tanh(\th_{uv})$ for all edges $(u,v) \in E$.
\end{lemma}
\begin{proof}
Fix any pair of nodes $(u,v)$. 
Let $\vec{\th}^p$ be the parameter vector of $p$, and $\vec{\th}^q$ be the parameter vector of an Ising model with all parameters equal to $0$, barring $\theta^q_{uv} = \theta^p_{uv}$.
The result follows by the expression for the edge marginal in $q$ in combination with Lemma~\ref{lem:ferro-percolation}.
\end{proof}

Given the above structural lemma about ferromagnetic Ising models under no external field, we present the following algorithm for testing whether a ferromagnetic Ising model is product or not.
\begin{algorithm}[htb]
\begin{algorithmic}[1]
\Function{TestFerroIsing-Independence}{sample access to an Ising model $p$}
\State \parbox[t]{\dimexpr\linewidth-\algorithmicindent}{Run the algorithm of Lemma \ref{lem:localization} to identify if all edges $e=(u,v)$ such that $\E[X_uX_v] \geq \sqrt{\ve}/n$ using $\tilde{O}\left(\frac{n^2}{\ve}\right)$ samples.
If it identifies any edges, return that $\dskl(p, \I) \geq \ve$.\strut}
\State Otherwise, return that $p$ is product.
\EndFunction
\end{algorithmic}
\caption{Test if a ferromagnetic Ising model $p$ under no external field is product}
\label{alg:ferro}
\end{algorithm}

\begin{prevproof}{Theorem}{thm:ferro-independence}
Firstly, note that under no external field, the only product Ising model is the uniform distribution $\U$. To the problem reduces to testing whether $p$ is uniform or not. Consider the case when $p$ is indeed uniform. That is, there are no edges in the underlying graph of the Ising model. In this case with probability at least $9/10$ the localization algorithm of Lemma \ref{lem:localization} with output no edges. Hence Algorithm \ref{alg:ferro} will output that $p$ is product.\\
In case $\dskl(p,\I) \geq \ve$, we split the analysis into two cases.
\begin{itemize}
\item \textit{Case 1:} There exists an edge $e=(u,v)$ such that $\abss{\th_{uv}} \geq \sqrt{\frac{\ve}{n^2}}$. In this case, $\abss{\E[X_uX_v]} \geq \abss{\tanh(\th_{uv})}$ and in the regime where $\ve$ is a fixed constant, $\abss{\tanh(\th)} \geq \abss{\th/2}$. Hence implying that $\abss{\E[X_uX_v]} \geq \abss{\th_{uv}/2} \geq \sqrt{\frac{\ve}{n^2}}/2$. Therefore the localization algorithm of Lemma \ref{lem:localization} would identify such an edge with probability at least $9/10$. (The regime where the inequality $\abss{\tanh(\th)} \geq \abss{\th/2}$ isn't valid would be easily detectable using $\tilde{O}(\frac{n^2}{\ve})$ samples.)
\item \textit{Case 2:} All edges $e=(u,v)$ are such that $\th_{uv} \leq \sqrt{\frac{\ve}{n^2}}$. 
(\ref{eq:dskl}) combined with this condition implies that there must exist an edge marginal of magnitude at least $\sqrt{\ve}/n$, and thus the localization algorithm of Lemma~\ref{lem:localization} is likely to identify it.
\end{itemize}
\end{prevproof}


\section{An Improved Test for High-Temperature Ising Models: A Learn-then-Test Algorithm}
\label{sec:learn-and-test}

In this section, we describe a framework for testing Ising models in the high-temperature regime which results in algorithms which are more efficient than our baseline localization algorithm of Section~\ref{sec:localization} for dense graphs. This is the more technically involved part of our paper and we modularize the description and analysis into different parts. We begin with a high level overview of our approach.

The main approach we take in this section is to consider a global test statistic over all the variables on the Ising model in contrast to the localized statistics of Section~\ref{sec:localization}. 
For ease of exposition, we first describe the approach for testing independence under no external field.
We then describe the changes that need to be made to obtain tests for independence under an external field and goodness-of-fit in Section~\ref{sec:changes}. 

Note that testing independence under no external field is the same as testing uniformity. 
The intuition for the core of the algorithm is as follows.
Suppose we are interested in testing uniformity of Ising model $p$ with parameter vector $\vec{\th}$.
We start by observing that, by (\ref{eq:dskl}), $\frac{\dskl(p,\U)}{\beta} \leq \sum_{u \neq v} \abss{\m_{uv}}$.
With this in mind, we consider the statistic $Z = \sum_{u \neq v} \sign(\m_{uv})\cdot \left(X_u X_v \right)$, where $X \sim p$ and $\sign(\m_{uv})$ is chosen arbitrarily if $\m_{uv} = 0$.
It is easy to see that $\E[Z] = \sum_{u \neq v} \abss{\m_{uv}}$.
If $X \in \I$, then $\E[Z] = 0$.
On the other hand, by the bound above, if $\dskl(X, \I) \geq \ve$ then $\E[Z] \geq \ve/\b$.
If the $\sign(\m_{uv})$ parameters were known, we could simply plug them into $Z$, and using Chebyshev's inequality, distinguish these two cases using $\Var(Z) \b^2/\ve^2$ samples. 

There are two main challenges here.
\begin{itemize}
\item First, the sign parameters, $\sign(\m_{uv})$, are \emph{not} known.
\item Second, it is not obvious how to get a non-trivial bound for $\Var(Z)$.
\end{itemize}

One can quickly see that learning all the sign parameters might be prohibitively expensive.
For example, if there is an edge $e$ such that $|\m_e| = 1/2^n$, there would be no hope of correctly estimating its sign with a polynomial number of samples. 
Instead, we perform a process we call \emph{weak learning} -- rather than trying to correctly estimate all the signs, we instead aim to obtain a $\vec \G$ which is \emph{correlated} with the vector $\sign(\m_e)$.
In particular, we aim to obtain $\vec \G$ such that, in the case where $\dskl(p, \U) \geq \ve$, $\E[\sum_{e = (u,v) \in E} \G_e \left(X_u X_v\right)] \geq \ve/\zeta \b$, where $\zeta = \poly(n)$. 
The main difficulty of analyzing this process is due to correlations between random variables $\left(X_uX_v\right)_{(u,v) \in E}$.
Naively, we could get an appropriate $\G_e$ for $\left(X_uX_v\right)$ by running a weak learning process independently for each edge.
However, this incurs a prohibitive cost of $O(n^2)$ by iterating over all edges.
We manage to sidestep this cost by showing that, despite these correlations, learning all $\G_e$ simultaneously succeeds with a probability which is $\geq 1/\poly(n)$, for a moderate polynomial in $n$.
Thus, repeating this process several times, we can obtain a $\vec \G$ which has the appropriate guarantee with sufficient constant probability.

At this point, we have a statistic 
\begin{equation}
Z' = \sum_{u \neq v} c_{uv} X_uX_v, \label{eq:statistic no external field}
\end{equation}
where $c \in \{\pm 1\}^{V \choose 2}$ represent the signs obtained from the weak learning procedure.
The gap in the expectation of $Z'$ in the two cases is $\ve/\zeta \b$, and thus by Chebyshev's inequality, they can be distinguished using $\Var(Z') \zeta^2 \b^2/\ve^2$ samples.
At this point, we run into the second issue mentioned above: we must bound the variance of $Z'$.
Since the range of $Z'$ is $\Omega(n^2)$, a crude bound for its variance is $O(n^4)$, granting us no savings over the localization algorithm of Theorem \ref{thm:localization}.
However, in the high temperature regime, Theorem \ref{thm:variance bound high temperature no external field yuval} shows the bound $\Var(Z') = O(n^2)$. 
In other words, despite the potentially complex structure of the Ising model and potential correlations, the variables $X_u X_v$ contribute to the variance of $Z'$ roughly as if they were all independent!
We describe the result and techniques involved in the analysis of the variance bound in Section \ref{sec:variance}.

The description of our algorithm for independence testing with no external field is presented in Algorithm \ref{alg:framework}. It depends on a parameter $\tau$, which can be set as $4/3$ to optimize the sample complexity as in Theorem~\ref{thm:learn-and-test}.
\begin{theorem}
\label{thm:learn-and-test}
Given $\tilde O\left(\min_{\t > 0} \left(n^{2 + \tau} + n^{6 - 2\tau} \right)\frac{\b^2}{\ve^2} \right) = \tilde O\left(\frac{n^{10/3}\beta^2}{\ve^2}\right)$ i.i.d.\ samples from an Ising model $p$ in the high-temperature regime with no external field, Algorithm~\ref{alg:framework} distinguishes between the cases $p \in \I$ and $\dskl(p, \I) \geq \ve$ with probability at least $2/3$.
\end{theorem}

\begin{algorithm}[htb]
\begin{algorithmic}[1]
\Function{Learn-Then-Test-Ising}{sample access to an Ising model $p,\b,\dm,\ve, \tau$}
\State \parbox[t]{\dimexpr\linewidth-\algorithmicindent}{Run the localization Algorithm \ref{alg:localization-product-test} on $p$ with accuracy parameter $\frac{\ve}{ n^{\t}}$.
If it identifies any edges, return that $\dskl(p, \I) \geq \ve$.\strut}
\For {$\ell = 1$ to $O(n^{2 - \tau})$}
\State \parbox[t]{\dimexpr\linewidth-\algorithmicindent}{Run the weak learning Algorithm \ref{alg:weak-learn} on $S = \{ X_uX_v\}_{u \neq v}$ with parameters $\tau$ and $\ve/\b$ to generate a sign vector $\vec \G^{(\ell)}$ where $\G^{(\ell)}_{uv}$ is weakly correlated with $\sign\left(\E\left[ X_{uv}\right]\right)$.\strut}
\EndFor
\State \parbox[t]{\dimexpr\linewidth-\algorithmicindent}{Using the \emph{same set of samples for all $\ell$}, run the testing algorithm of Lemma \ref{lem:chebytest} on each of the $\vec \G^{(\ell)}$ with parameters $\tau_2 = \tau, \d = O(1/n^{2 - \tau})$. If any output that $\dskl(p, \I) \geq \ve$, return that $\dskl(p, \I) \geq \ve$. Otherwise, return that $p \in \I$.\strut}
\EndFunction
\end{algorithmic}
\caption{Test if an Ising model $p$ under no external field is product using Learn-Then-Test}
\label{alg:framework}
\end{algorithm}



\begin{remark}
The first step in Algorithm \ref{alg:framework} is to perform a simple localization test to check if $\abss{\m_e}$ is not too far away from $0$ for all $e$. 
It is added to help simplify the analysis of the algorithm and is not necessary in principle.
In particular, if we pass this test, then the rest of the algorithm has the guarantee that $|\m_e|$ is small for all $e \in E$.
\end{remark}





The organization of the rest of the section is as follows.
We describe and analyze our weak learning procedure in Section \ref{sec:weak-learning}.
Given a vector with the appropriate weak learning guarantees, we describe and analyze the testing procedure in Section \ref{sec:hyp-testing}. 
In Section \ref{sec:combine}, we describe how to combine all these ideas -- in particular, our various steps have several parameters, and we describe how to balance the complexities to obtain the sample complexity stated in Theorem \ref{thm:learn-and-test}.

\subsection{Weak Learning}
\label{sec:weak-learning}

Our overall goal of this section is to ``weakly learn'' the sign of $\m_e = \E[X_uX_v]$ for all edges $e = (u,v)$.
More specifically, we wish to output a vector $\vec \G$ with the guarantee that 
$\E_{X}\left[\sum_{e = (u,v) \in E} \G_e X_u X_v \right] \geq \frac{c \ve}{2\b n^{2 - \tau_2}}$,
for some constant $c > 0$ and parameter $\tau_2$ to be specified later.
Note that the ``correct'' $\G$, for which $\G_e = \sign(\m_e)$, has this guarantee with $\tau_2 = 2$ -- by relaxing our required learning guarantee, we can reduce the sample complexity in this stage.

The first step will be to prove a simple but crucial lemma answering the following question: 
Given $k$ samples from a Rademacher random variable with parameter $p$, how well can we estimate the sign of its expectation?
This type of problem is well studied in the regime where $k = \Omega(1/p^2)$, in which we have a constant probability of success, but we analyze the case when $k \ll 1/p^2$ and prove how much better one can do versus randomly guessing the sign.
See Lemma \ref{lem:weakLearningBernoulli} in Section \ref{sec:weak-learning-bern} for more details.

With this lemma in hand, we proceed to describe the weak learning procedure. Given parameters $\t,\ve$ and sample access to a set $S$ of `Rademacher-like' random variables which may be \emph{arbitrarily correlated} with each other, the algorithm draws $\tilde{O}\left(\frac{n^{2\t}}{\ve^2} \right)$ samples from each random variable in the set and computes their empirical expected values and outputs a signs of thus obtained empirical expectations. The procedure is described in Algorithm \ref{alg:weak-learn}.

\begin{algorithm}[htb]
\begin{algorithmic}[1]
\Function{WeakLearning}{sample access to set $\{ Z_{i} \}_{i}$ of random variables where $Z_{i} \in \{\pm 1\}$ and can be arbitrarily correlated, $\ve$, $\t$}
\State Draw $k = \tilde{O}\left(\frac{n^{2\t}}{\ve^2} \right)$ samples from each $Z_{i}$. Denote the samples by $Z_{i}^{(1)},\ldots,Z_{i}^{(k)}$.
\State Compute the empirical expectation for each $Z_{i}$: $\hat{Z}_{i} = \frac{1}{k}\sum_{l=1}^k Z_{i}^{(l)}$.
\State Output $\vec \G$ where $\G_{i} = \sign(\hat{Z}_{i})$.
\EndFunction
\end{algorithmic}
\caption{Weakly Learn Signs of the Expectations of a set of Rademacher-like random variables}
\label{alg:weak-learn}
\end{algorithm}

We now turn to the setting of the Ising model, discussed in Section \ref{sec:weak-learning-ising}. We invoke the weak-learning procedure of Algorithm \ref{alg:weak-learn} on the set $S = \{ X_uX_v\}_{u \neq v}$ with parameters $\ve/\b$ and $0 \le \t \le 2$.
By linearity of expectations and Cauchy-Schwarz, it is not hard to see that we can get a guarantee of the form we want in expectation (see Lemma \ref{lem:largeExpectation}).
However, the challenge remains to obtain this guarantee with constant probability.
Carefully analyzing the range of the random variable and using this guarantee on the expectation allows us to output an appropriate vector $\vec \G$ with probability inversely polynomial in $n$ (see Lemma \ref{lem:weak-learning-prob-bound}).
Repeating this process several times will allow us to generate a collection of candidates $\{\vec \G^{(\ell)}\}$, at least one of which has our desired guarantees with constant probability.

\subsubsection{Weak Learning the Edges of an Ising Model}
\label{sec:weak-learning-ising}

We now turn our attention to weakly learning the edge correlations in the Ising model.
To recall, our overall goal is to obtain a vector $\vec \G$ such that $\E_{X \sim p}\left[\sum_{e = (u,v) \in E} \G_e X_u X_v \right] \geq \frac{c \ve}{2\b n^{2 - \tau_2}}$.

We start by proving that Algorithm \ref{alg:weak-learn} yields a $\vec \G$ for which such a bound holds in expectation.
The following is fairly straightforward from Lemma \ref{lem:weakLearningBernoulli} and linearity of expectations.
\begin{lemma}
\label{lem:largeExpectation}
Given $k = O\left(\frac{n^{2\tau_2} \beta^2}{\ve^2}\right)$ samples from an Ising model $X \sim p$ such that $\dskl(p, \I) \geq \ve$ and $|\m_e| \leq \frac{\ve}{\b n^{\tau_2}}$ for all $e \in E$, Algorithm \ref{alg:weak-learn} outputs $\vec \G = \{\G_e\} \in \{\pm 1\}^{|E|}$ such that
$$\E_{\vec \G} \left[ \E_{X \sim p} \left[ \sum_{e = (u,v) \in E} \G_e X_u X_v \right] \right] \geq \frac{c\beta}{\ve n^{2 - \tau_2}}\left(\sum_{e \in E} |\m_e|\right)^2,$$
for some constant $c > 0$.
\end{lemma}
\begin{proof}
Since for all $e = (u,v) \in E$, $|\m_e|\leq \frac{\ve}{\beta n^{\tau_2}}$, and by our upper bound on $k$, all of the random variables $X_u X_v$ fall into the first case of Lemma \ref{lem:weakLearningBernoulli} (the ``small $k$'' regime).
Hence, we get that $\Pr \left[ \G_e=\sign(\m_e) \right] \geq \frac{1}{2}+ \frac{c_1|\m_e|\sqrt{k}}{2}
$ which implies that $\E_{\G_e} \left[ \G_e \m_e \right] \geq \left(\frac{1}{2}+ \frac{c_1|\m_e|\sqrt{k}}{2}\right)|\m_e| + \left(\frac{1}{2}- \frac{c_1|\m_e|\sqrt{k}}{2}\right)(-|\m_e|)
= c_1|\m_e|^2\sqrt{k}$.
Summing up the above bound over all edges, we get
$$\E_{\vec{\G}}\left[ \sum_{e \in E} \G_e \m_e\right] \geq c_1\sqrt{k} \sum_{e \in E} |\m_e|^2 
\geq \frac{c_1' n^{\tau_2}\beta }{\ve}\sum_{e \in E} |\m_e|^2$$
for some constant $c_1' > 0$.
Applying the Cauchy-Schwarz inequality gives us
$\E_{\vec \G }\left[ \sum_{e \in E} \G_e \m_e\right] \geq \frac{c\beta }{\ve n^{2-\tau_2}}\left(\sum_{e \in E} |\m_e|\right)^2$ as desired.
\end{proof}

Next, we prove that the desired bound holds with sufficiently high probability.
The following lemma follows by a careful analysis of the extreme points of the random variable's range. 

\begin{lemma}
\label{lem:weak-learning-prob-bound}
Define $\goodgamma$ to be the event that $\vec \G = \{\G_e\} \in \{\pm 1\}^{|E|}$ is such that $$\E_{X \sim p}\left[\sum_{e = (u,v) \in E} \G_e X_u X_v\right] \geq \frac{c\ve}{2\b n^{2 - \tau_2}},$$
for some constant $c > 0$.
Given $k = O\left(\frac{n^{2\tau_2}\beta^2}{\ve^2}\right)$ i.i.d.\ samples from an Ising model $p$ such that $\dskl(p, \I) \geq \ve$ and $|\m_e| \leq \frac{\ve}{\b n^{\tau_2}}$ for all $e \in E$, Algorithm \ref{alg:weak-learn} outputs $\vec \G$ satisfying $\goodgamma$ with probability at least $\frac{c}{4 n^{2 - \tau_2}}.$
\end{lemma}
\begin{proof}
We introduce some notation which will help in the elucidation of the argument which follows. Let $r$ be the probability that the $\vec \G$ output by Algorithm~\ref{alg:weak-learn} satisfies $\goodgamma$.
Let $T = \frac{c \beta}{2\ve n^{2 - \tau_2}}\left(\sum_{e \in E} |\m_e|\right)^2$.
Let $Y,U,L$ be random variables defined as follows: 
$$ Y = \E_{X\sim p}\left[\sum_{e = (u,v) \in E} \G_e X_u X_v\right],\ 
 U = \E_{\vec \G} \left[ Y | Y > T \right],\ 
 L = \E_{\vec \G} \left[ Y | Y \leq T \right].$$

Then by Lemma~\ref{lem:largeExpectation}, we have $rU + (1-r)L \ge 2T$, which implies that $r \ge \frac{2T-L}{U-L}$.
Since $- \sum_{e \in E} |\m_e| \leq L \leq T \leq U \leq \sum_{e \in E} |\m_e|$, we have $r \geq \frac{T}{2\left(\sum_{e \in E} |\m_e|\right)}$.
Substituting in the value for $T$ we get $r \geq \frac{c\b\left(\sum_{e \in E} |\m_e|\right)^2}{4\ve n^{2-\tau_2}\left(\sum_{e \in E} |\m_e|\right)} = \frac{c\b\left(\sum_{e \in E} |\m_e|\right)}{4\ve n^{2-\tau_2}}$.
Since $\dskl(p, \I) \geq \ve$, this implies $\left(\sum_{e \in E} |\m_e|\right) \geq \ve/\b$ and thus
$ r \geq \frac{c}{4n^{2-\tau_2}}, $
as desired.
\end{proof}

\subsection{Testing Our Learned Hypothesis}
\label{sec:hyp-testing}
In this section, we assume that we were successful in weakly learning a vector $\vec \Gamma$ which is ``good'' (i.e., it satisfies $\goodgamma$, which says that the expectation the statistic with this vector is sufficiently large).
With such a $\vec \Gamma$, we show that we can distinguish between $p \in \I$ and $\dskl(p, \I) \geq \ve$.
\begin{lemma}
\label{lem:chebytest}
Let $p$ be an Ising model, let $X \sim p$, and let $\s^2$ be such that, for any $\vec \g = \{\g_e\} \in \{\pm 1\}^{|E|}$,
$\Var\left(\sum_{e = (u,v) \in E} \g_e X_u X_v \right) \leq \s^2 .$
Given $k = O\left(\s^2 \cdot \frac{n^{4 - 2\tau_2 }\beta^2 \log(1/\d)}{\ve^2}\right)$ i.i.d.\ samples from $p$, which satisfies either $p \in \I$ or $\dskl(p, \I) \geq \ve$, and $\vec \G = \{\G_e\} \in \{\pm 1\}^{|E|}$ which satisfies $\goodgamma$ (as defined in Lemma \ref{lem:weak-learning-prob-bound}) in the case that $\dskl(p, \I) \geq \ve$, then there exists an algorithm which distinguishes these two cases with probability $\geq 1 - \d$.
\end{lemma}
\begin{proof}
We prove this lemma with failure probability $1/3$ -- by standard boosting arguments, this can be lowered to $\d$ by repeating the test $O(\log (1/\d))$ times and taking the majority result.

Denoting the $i$th sample as $X^{(i)}$, the algorithm computes the statistic 
$$Z = \frac{1}{k} \left(\sum_{i=1}^{k} \sum_{e=(u,v) \in E} \G_e X^{(i)}_u X^{(i)}_v\right).$$
If $Z \leq \frac{c \ve}{4 \b n^{2 - \tau_2}}$, then the algorithm outputs that $p \in \I$, otherwise, it outputs that $\dskl(p, \I) \geq \ve$.

By our assumptions in the lemma statement $ \Var\left(Z\right) \leq \frac{\s^2}{k}$.
If $p \in \I$, then we have that $\E[Z] = 0$, and Chebyshev's inequality implies that 
$ \Pr\left[Z \geq \frac{\ve}{4\beta n^{2-\tau_2}} \right] \leq \frac{16\s^2\beta^2 n^{4-2\tau_2}}{kc^2\ve^2}$.
Substituting the value of $k$ gives the desired bound in this case.
The case where $\dskl(p, \I) \geq \ve$ follows similarly, using the fact that $\goodgamma$ implies 
$\E[Z] \geq \frac{c \ve}{2 \b n^{2 - \tau_2}}$ .
\end{proof}

\subsection{Combining Learning and Testing}
\label{sec:combine}
In this section, we combine lemmas from the previous sections to complete the proof of Theorem \ref{thm:learn-and-test}.
Lemma \ref{lem:weak-learning-prob-bound} gives us that a single iteration of the weak learning step gives a ``good'' $\vec \G$ with probability at least $\Omega\left(\frac{1}{n^{2-\tau_2}}\right)$. 
We repeat this step $O(n^{2 - \tau_2})$ times, generating $O(n^{2 - \tau_2})$ hypotheses $\vec \G^{(\ell)}$.
By standard tail bounds on geometric random variables, this will imply that at least one hypothesis is good (i.e. satisfying $\goodgamma$) with probability at least $9/10$.
We then run the algorithm of Lemma \ref{lem:chebytest} on each of these hypotheses, with failure probability $\d = O(1/n^{2 - \tau_2})$.
If $p \in \I$, all the tests will output that $p \in \I$ with probability at least $9/10$.
Similarly, if $\dskl(p, \I) \geq \ve$, conditioned on at least one hypothesis $\vec \G^{(\ell^*)}$ being good, the test will output that $\dskl(p, \I) \geq \ve$ for this hypothesis with probability at least $9/10$.
This proves correctness of our algorithm.

To conclude our proof, we analyze its sample complexity.
Combining the complexities of Lemmas \ref{lem:localization}, \ref{lem:weak-learning-prob-bound}, and \ref{lem:chebytest}, the overall sample complexity is
$
O\left(\frac{n^{2\tau_1}\beta^2\log n}{\ve^2}\right)+O\left(\frac{n^{2+\tau_2}\beta^2}{\ve^2} \right)+O\left( \s^2\frac{n^{4-2\tau_2}\beta^2}{\ve^2} \log n\right).
$
Noting that the first term is always dominated by the second term we can simplify the complexity to the expression $O\left(\frac{n^{2+\tau_2}\beta^2}{\ve^2} \right)+O\left( \s^2\frac{n^{4-2\tau_2}\beta^2}{\ve^2} \log n\right)$.
Plugging in the variance bound from Section~\ref{sec:variance} (Theorem \ref{thm:variance bound high temperature no external field yuval}) gives Theorem \ref{thm:learn-and-test}.

\subsection{Changes Required for General Independence and Identity Testing}
\label{sec:changes}
We describe the modifications that need to be done to the learn-then-test approach described in Sections \ref{sec:weak-learning}-\ref{sec:combine} to obtain testers for independence under an arbitrary external field (Section~\ref{sec:changes-ind}), identity without an external field (Section~\ref{sec:changes-id-nofield}), and identity under an external field (Section~\ref{sec:changes-id-field}).

\subsubsection{Independence Testing under an External Field}
\label{sec:changes-ind}
Under an external field, the statistic we considered in Section \ref{sec:learn-and-test} needs to be modified. 
We are interested in testing independence of an ising model $p$ on graph $G = (V,E)$.
We have that $\dskl(p,\I) = \min_{q \in \I} \dskl(p,q)$. 
In particular, let $q$ be the product Ising model on graph $G'=(V,\emptyset)$ where $\m_u^q = \m_u^p$ for all $u \in V$. 
Manipulation of (\ref{eq:dskl}) gives that $\frac{\dskl(p,\I)}{\b} \leq  \sum_{e=(u,v) \in E} \abss{\m_{uv}^p - \m_u^p\m_v^p}$.
This suggests a statistic $Z$ such that $\E[Z] = \sum_{e=(u,v) \in E} \abss{\l_{uv}^p}$ where $\l_{uv}^p = \m_{uv}^p - \m_u^p\m_v^p$. 
We consider 
$$Z = \frac12 \sum_{u \neq v} \sign(\l_{uv})\left(X_u^{(1)}-X_u^{(2)} \right)\left(X_v^{(1)}-X_v^{(2)} \right),$$
 where $X^{(1)},X^{(2)} \sim p$ are two independent samples from $p$. 
It can be seen that $Z$ has the desired expectation. 
We face the same challenge as before, since we do not know the $\sign(\l_{uv})$ parameters, so we again apply our weak learning procedure.
Consider the following random variable: 
$Z_{uv} = \frac{1}{4}\left(X_u^{(1)} - X_u^{(2)}\right)\left(X_v^{(1)} - X_v^{(2)}\right).$
Though $Z_{uv}$ takes values in $\{-1,0,+1 \}$, we can easily transform it to the domain $\{\pm 1\}$; define $Z'_{uv}$ to be the Rademacher random variable which is equal to $Z_{uv}$ if $Z_{uv} \neq 0$, and is otherwise chosen to be $\{\pm 1\}$ uniformly at random.
It is easy to see that this preserves the expectation: $\E[Z'_{uv}] = \E[Z_{uv}] = \frac{\l_{uv}}{2}$. 
Hence $Z'_{uv} \sim Rademacher\left(\frac{1}{2} + \frac{\l_{uv}}{4} \right)$, and given $k$ samples, Lemma \ref{lem:weakLearningBernoulli} allows us to learn its sign correctly with probability at least $1/2+c_1 \sqrt{k}\abss{\l_{uv}}$.
The rest of the weak learning argument of Lemmas \ref{lem:largeExpectation} and \ref{lem:weak-learning-prob-bound} follows by replacing $\m_e$ with $\l_e$.

Once we have \emph{weakly learnt} the signs, we are left with a statistic $Z'_{cen}$ of the form:
\begin{align}
Z'_{cen} = \sum_{u \neq v} c_{uv} \left(X_u^{(1)}-X_u^{(2)} \right)\left(X_v^{(1)}-X_v^{(2)}\right). \label{eq:statistic external field}
\end{align}

We need to obtain a bound on $\Var(Z'_{cen})$. 
We again employ the techniques described in Section \ref{sec:variance} to obtain a non-trivial bound on $\Var(Z'_{cen})$ in the high-temperature regime. 
The statement of the variance result is given in Theorem \ref{thm:variance bound high temperature with external field yuval} and the details are in Section \ref{sec:variance bound external field yuval}.
Putting this all together gives us the sample complexity 
$ \tilde{O}\left(\frac{(n^{2+\t} + n^{4-2\t}n^2)\b^2}{\ve^2} \right)$, where we again choose $\tau = 4/3$ to obtain Theorem~\ref{thm:learn-then-test-ind-extfield-balanced}.
The algorithm is described in Algorithm~\ref{alg:independence-learn-then-test-external-field}.

\begin{theorem}[Independence Testing using Learn-Then-Test, Arbitrary External Field]
\label{thm:learn-then-test-ind-extfield-balanced}
Suppose $p$ is an Ising model in the high temperature regime under an arbitrary external field. The learn-then-test algorithm takes in $\tilde{O}\left( \frac{n^{10/3}\b^2}{\ve^2}\right)$ i.i.d.\ samples from $p$ and distinguishes between the cases $p \in \I$ and $\dskl(p,\I) \ge \ve$ with probability $\ge 9/10$.
\end{theorem}

\begin{algorithm}[htb]
\begin{algorithmic}[1]
\Function{Learn-Then-Test-Ising}{sample access to an Ising model $p,\b,\dm,\ve, \tau$}
\State \parbox[t]{\dimexpr\linewidth-\algorithmicindent}{Run the localization Algorithm \ref{alg:localization-product-test} with accuracy parameter $\frac{\ve}{n^{\t}}$.
If it identifies any edges, return that $\dskl(p, \I) \geq \ve$.\strut} 
\For {$\ell = 1$ to $O(n^{2 - \tau})$}
\State \parbox[t]{\dimexpr\linewidth-\algorithmicindent}{Run the weak learning Algorithm \ref{alg:weak-learn} on $S = \{ (X_u^{(1)}-X_u^{(2)})(X_v^{(1)}-X_v^{(2)})\}_{u \neq v}$ with parameters $\tau_2 = \tau$ and $\ve/\b$ to generate a sign vector $\vec \G^{(\ell)}$ where $\G^{(\ell)}_{uv}$ is weakly correlated with $\sign\left(\E\left[(X_u^{(1)}-X_u^{(2)})(X_v^{(1)}-X_v^{(2)})\right]\right)$.\strut}
\EndFor
\State \parbox[t]{\dimexpr\linewidth-\algorithmicindent}{Using the \emph{same set of samples for all $\ell$}, run the testing algorithm of Lemma \ref{lem:chebytest} on each of the $\vec \G^{(\ell)}$ with parameters $\tau_2 = \tau, \d = O(1/n^{2 - \tau})$. If any output that $\dskl(p, \I) \geq \ve$, return that $\dskl(p, \I) \geq \ve$. Otherwise, return that $p \in \I$.\strut}
\EndFunction
\end{algorithmic}
\caption{Test if an Ising model $p$ under arbitrary external field is product}
\label{alg:independence-learn-then-test-external-field}
\end{algorithm}

\subsubsection{Identity Testing under No External Field}
\label{sec:changes-id-nofield}
We discuss the changes needed for identity testing under no external field. 
Similar to before, we start by upper bounding the $\dskl$ between the Ising models $p$ and $q$, obtaining
$\frac{\dskl(p,q)}{2\b} \le \sum_{u \neq v} \abss{\m_{uv}^p-\m_{uv}^q }$.
Since we know $\m_{uv}^q$ for all pairs $u,v$, this suggests a statistic of the form
\begin{align*}
Z = \sum_{u \neq v} \sign\left( \m_{uv}^p - \m_{uv}^q\right)\left(X_uX_v - \m_{uv}^q \right).
\end{align*}
By separating out the part of the statistic which is just a constant, we obtain that
$\Var(Z) \le \Var\left(\sum_{u \neq v} c_{uv} X_uX_v \right),$
 and we can apply the variance bound of Theorem \ref{thm:variance bound high temperature no external field yuval}.

In this case, we can weakly learn the signs using Corollary~\ref{cor:weakLearningBernoulliGeneral} of Lemma~\ref{lem:weakLearningBernoulli}.
Given $k$ samples, we correctly learn the sign of $\m_{uv}^p - \m_{uv}^q$ with probability at least $1/2+c_1\sqrt{k}\abss{\m_{uv}^p - \m_{uv}^q}$.
After replacing $\m_e$ by $\m_e^p - \m_e^q$, we can prove statements analogous to Lemmas~\ref{lem:largeExpectation} and \ref{lem:weak-learning-prob-bound}.
These changes allow us to conclude the following theorem, which formalizes the guarantees of Algorithm~\ref{alg:identity-learn-then-test-no-field}.
\begin{theorem}[Identity Testing using Learn-Then-Test, No External Field]
\label{thm:learn-then-test-id-noextfield-balanced}
Suppose $p$ and $q$ are Ising models in the high temperature regime under no external field. The learn-then-test algorithm takes in $\tilde{O}\left(\frac{n^{10/3}\b^2}{\ve^2} \right)$ i.i.d.\ samples from $p$ and distinguishes between the cases $p = q$ and $\dskl(p,q) \ge \ve$ with probability $\ge 9/10$.
\end{theorem}



\begin{algorithm}[htb]
\begin{algorithmic}[1]
\Function{TestIsing}{sample access to an Ising model $p,\b,\dm,\ve, \tau$, description of Ising model $q$ under no external field}
\State \parbox[t]{\dimexpr\linewidth-\algorithmicindent}{Run the localization Algorithm \ref{alg:localization-identity-test} with accuracy parameter $\frac{\ve}{n^{\t}}$.
If it identifies any edges, return that $\dskl(p, q) \geq \ve$.\strut} 
\For {$\ell = 1$ to $O(n^{2 - \tau})$}
\State \parbox[t]{\dimexpr\linewidth-\algorithmicindent}{Run the weak learning Algorithm \ref{alg:weak-learn} on $S = \{ X_uX_v - \m_{uv}^q\}_{u \neq v}$ with parameters $\tau_2 = \tau$ and $\ve/\b$ to generate a sign vector $\vec \G^{(\ell)}$ where $\G^{(\ell)}_{uv}$ is weakly correlated with $\sign\left(\E\left[ X_{uv} - \m_{uv}^q\right]\right)$.\strut}
\EndFor
\State \parbox[t]{\dimexpr\linewidth-\algorithmicindent}{Using the \emph{same set of samples for all $\ell$}, run the testing algorithm of Lemma \ref{lem:chebytest} on each of the $\vec \G^{(\ell)}$ with parameters $\tau_2 = \tau, \d = O(1/n^{2 - \tau})$. If any output that $\dskl(p, q) \geq \ve$, return that $\dskl(p, q) \geq \ve$. Otherwise, return that $p  = q$.\strut}
\EndFunction
\end{algorithmic}
\caption{Test if an Ising model $p$ under no external field is identical to $q$}
\label{alg:identity-learn-then-test-no-field}
\end{algorithm}

\subsubsection{Identity Testing under an External Field}
\label{sec:changes-id-field}
There are two significant differences when performing identity testing under an external field.
First, we must now account for the contributions of nodes in (\ref{eq:dskl}).
Second, it is not clear how to define an appropriately centered statistic which has a variance of $O(n^2)$ in this setting, and we consider this an interesting open question. 
Instead, we use the uncentered statistic which has variance $\Theta(n^3)$. 

Again, we start by considering an upper bound on the SKL between Ising models $p$ and $q$.
$\dskl(p,q) \le 2h \sum_{v \in V} \abss{\m_v^p - \m_v^q } + 2\b \sum_{u \neq v} \abss{\m_{uv}^p - \m_{uv}^q }$, and thus if $\dskl(p,q) \ge \ve$, then either 
$2h\sum_{v \in V} \abss{\m_v^p - \m_v^q } \ge \ve/2$ or
$2\b \sum_{u \neq v} \abss{\m_{uv}^p - \m_{uv}^q } \ge \ve/2$.
Our algorithm will first check the former case, to see if the nodes serve as a witness to $p$ and $q$ being far. 
If they do not, we proceed to check the latter case, checking if the edges cause $p$ and $q$ to be far.

We first describe the test to detect whether $\sum_{v \in V} \abss{\m_v^p - \m_v^q } = 0$ or is $\ge \ve/4h$. 
We consider a statistic of the form $Z = \sum_{v \in V} \sign(\m_v^p)\left(X_v - \m_v^q \right)$. 
We employ the weak-learning framework described in Sections \ref{sec:weak-learning} to weakly learn a sign vector correlated with the true sign vector. 
Since $X_v \sim Rademacher(1/2+\m_v/2)$, Corollary \ref{cor:weakLearningBernoulliGeneral} implies that with $k$ samples, we can correctly estimate $\sign(\m_v^p - \m_v^q)$ with probability $1/2 + c_1\sqrt{k}\abss{\m_v^p - \m_v^q}$. 
The rest of the argument follows similar to before, though we enjoy some savings since we only have a linear number of nodes (compared to a quadratic number of edges), corresponding to a linear number of terms.
Letting $f_c(X) = \sum_{v \in V} c_v X_v$ for some sign vector $c$, Lemma~\ref{lem:lipschitz-lemma} implies $\Var(f_c(X)) = O(n)$.
By calculations analogous to the ones in Sections \ref{sec:combine}, we obtain that by with $\tilde{O}\left(\frac{n^{5/3}h^2}{\ve^2} \right)$ samples, we can test whether $\sum_{v \in V} \abss{\m_v^p - \m_v^q } = 0$ or $\ge \ve/4h$ with probability $\ge 19/20$. 
If the tester outputs that $\sum_{v \in V} \abss{\m_v^p - \m_v^q} = 0$, then we proceed to test whether $\sum_{u \neq v} \abss{\m_{uv}^p - \m_{uv}^q } = 0$ or $\ge \ve/4\b$. 

To perform this step, we begin by looking at the statistic $Z$ used in Section \ref{sec:changes-id-nofield}:
$$Z = \sum_{u \neq v} \sign\left( \m_{uv}^p - \m_{uv}^q\right)\left(X_uX_v - \m_{uv}^q \right).$$
Note that $\E[Z] = \sum_{u \neq v} |\mu_{uv}^p - \mu_{uv}^q|$.
As before, we apply weak learning to obtain a sign vector which is weakly correlated with the true signs.
We also need a variance bound on functions of the form $f_c(X) = \sum_{u \neq v} c_{uv} (X_uX_v - \m_{uv}^q)$ where $c$ is a sign vector. 
By ignoring the constant term in $f_c(X)$, we get that
$\Var(f_c(X)) = \Var\left( \sum_{u \neq v} c_{uv} X_uX_v \right)$.
Note that this quantity can be $\Omega(n^3)$, as it is not appropriately centered. 
We employ Lemma \ref{lem:lipschitz-lemma} to get a variance bound of $O(n^3)$ which yields a sample complexity of $\tilde{O}\left(\frac{n^{11/3}\b^2}{\ve^2} \right)$ for this setting.
Our algorithm is presented as Algorithm~\ref{alg:identity-learn-then-test-external-field}, and its guarantees are summarized in Theorem~\ref{thm:learn-then-test-id-extfield-balanced}.
\begin{theorem}[Identity Testing using Learn-Then-Test, Arbitrary External Field]
\label{thm:learn-then-test-id-extfield-balanced}
Suppose $p$ and $q$ are Ising models in the high temperature regime under arbitrary external fields. The learn-then-test algorithm takes in $\tilde{O}\left(\frac{n^{5/3}h^2 + n^{11/3}\b^2}{\ve^2} \right)$ i.i.d.\ samples from $p$ and distinguishes between the cases $p =q$ and $\dskl(p,q) \ge \ve$ with probability $\ge 9/10$.
\end{theorem}


\begin{algorithm}[htb]
\begin{algorithmic}[1]
\Function{TestIsing}{sample access to an Ising model $p,\b,\dm,\ve, \t_1,\t_2$, description of Ising model $q$}
\State \parbox[t]{\dimexpr\linewidth-\algorithmicindent}{Run the localization Algorithm \ref{alg:localization-identity-test} on the nodes with accuracy parameter $\frac{\ve}{2n^{\t_1}}$.
If it identifies any nodes, return that $\dskl(p, q) \geq \ve$.\strut}
\For {$\ell = 1$ to $O(n^{1 - \tau_1})$}
\State \parbox[t]{\dimexpr\linewidth-\algorithmicindent}{Run the weak learning Algorithm \ref{alg:weak-learn} on $S = \{ (X_u - Y_{u}\}_{u \in V}$, where $Y_{u} \sim Rademacher(1/2+\m_{u}^q/2)$, with parameters $\tau_1$ and $\ve/2h$ to generate a sign vector $\vec \G^{(\ell)}$ where $\G^{(\ell)}_{u}$ is weakly correlated with $\sign\left(\E\left[X_u - \m_{u}^q\right]\right)$.\strut}
\EndFor
\State \parbox[t]{\dimexpr\linewidth-\algorithmicindent}{Using the \emph{same set of samples for all $\ell$}, run the testing algorithm of Lemma \ref{lem:chebytest} on each of the $\vec \G^{(\ell)}$ with parameters $\t_3 = \t_1, \d = O(1/n^{1 - \tau_1})$. If any output that $\dskl(p, q) \geq \ve$, return that $\dskl(p, q) \geq \ve$.\strut}
\State --------------------------
\State \parbox[t]{\dimexpr\linewidth-\algorithmicindent}{Run the localization Algorithm \ref{alg:localization-identity-test} on the edges with accuracy parameter $\frac{\ve}{2n^{\t_2}}$.
If it identifies any edges, return that $\dskl(p, q) \geq \ve$.\strut}
\For {$\ell = 1$ to $O(n^{2 - \tau_2})$}
\State \parbox[t]{\dimexpr\linewidth-\algorithmicindent}{Run the weak learning Algorithm \ref{alg:weak-learn} on $S = \{ (X_uX_v - Y_{uv}\}_{u \neq v}$, where $Y_{uv} \sim Rademacher(1/2+\m_{uv}^q/2)$, with parameters $\tau_2$ and $\ve/2\b$ to generate a sign vector $\vec \G^{(\ell)}$ where $\G^{(\ell)}_{uv}$ is weakly correlated with $\sign\left(\E\left[X_uX_v - \m_{uv}^q\right]\right)$.\strut}
\EndFor
\State \parbox[t]{\dimexpr\linewidth-\algorithmicindent}{Using the \emph{same set of samples for all $\ell$}, run the testing algorithm of Lemma \ref{lem:chebytest} on each of the $\vec \G^{(\ell)}$ with parameters $\t_4 = \t_2, \d = O(1/n^{2 - \tau_2})$. If any output that $\dskl(p, q) \geq \ve$, return that $\dskl(p, q) \geq \ve$. Otherwise, return that $p = q$.}
\EndFunction
\end{algorithmic}
\caption{Test if an Ising model $p$ under an external field is identical to Ising model $q$}
\label{alg:identity-learn-then-test-external-field}
\end{algorithm}

\section{Improved Tests for High-Temperature Ferromagnetic Ising Models}
\label{sec:htferro}
In this section, we present an improved upper bound for testing uniformity of Ising models which are both high-temperature and ferromagnetic.
Similar to the algorithms of Section~\ref{sec:learn-and-test}, we use a \emph{global} statistic, in comparison to the local statistic which is employed for general ferromagnets in Section~\ref{sec:ferro}.

Our result is the following:
\begin{theorem}[Independence Testing of High-Temperature Ferromagnetic Ising Models]
\label{thm:htferro-independence}
Algorithm~\ref{alg:htferro} takes in $\tilde{O}\left( \frac{n}{\ve}\right)$ samples from a high-temperature ferromagnetic Ising model $X \sim p$ which is under no external field and outputs whether $p \in \I$ or $\dskl(p,\I) \geq \ve$ with probability $\geq 9/10$.
\end{theorem}

We note that a qualitatively similar result was previously shown in~\cite{GheissariLP18}, using a $\chi^2$-style statistic.
Our algorithm is extremely similar to our test for general high-temperature Ising models.
The additional observation is that, since the model is ferromagnetic, we know that all edge marginals have non-negative expectation, and thus we can skip the ``weak learning'' stage by simply examining the global statistic with the all-ones coefficient vector.
The test is described precisely in Algorithm~\ref{alg:htferro}.

\begin{algorithm}[htb]
\begin{algorithmic}[1]
\Function{TestHighTemperatureFerroIsing-Independence}{sample access to an Ising model $p$}
\State \parbox[t]{\dimexpr\linewidth-\algorithmicindent}{Run the algorithm of Lemma \ref{lem:localization} to identify if all edges $e=(u,v)$ such that $\E[X_uX_v] \geq \sqrt{\ve/n}$ using $\tilde{O}\left(\frac{n}{\ve}\right)$ samples. \label{st:filter}
If it identifies any edges, return that $\dskl(p, \I) \geq \ve$.\strut}
\State Draw $k = O\left(\frac{n}{\ve}\right)$ samples from $p$, denote them by $X^{(1)}, \dots, X^{(k)}$.
\State Compute the statistic $Z = \frac{1}{k}\sum_{i=1}^k \sum_{(u,v) \in E} X^{(i)}_u X^{(i)}_v$.
\State If $Z \geq \frac{1}{4}\sqrt{\ve n}$, return that $\dskl(p, \I) \geq \ve$.
\State Otherwise, return that $p$ is product.
\EndFunction
\end{algorithmic}
\caption{Test if a high-temperature ferromagnetic Ising model $p$ under no external field is product}
\label{alg:htferro}
\end{algorithm}

\begin{prevproof}{Theorem}{thm:htferro-independence}
First, note that under no external field, the only product Ising model is the uniform distribution $\U$, and the problem reduces to testing whether $p$ is uniform or not. 
Consider first the filtering in Step~\ref{st:filter}.
By the correctness of Lemma~\ref{lem:localization}, this will not wrongfully reject any uniform Ising models.
Furthermore, for the remainder of the algorithm, we have that $\E[X_uX_v] \leq \sqrt{\ve/n}$.

Now, we consider the statistic $Z$.
By Theorem~\ref{thm:variance bound high temperature no external field yuval}, we know that the variance of $Z$ is at most $O\left(n^2/k\right)$ (since we are in high-temperature).
It remains to consider the expectation of the statistic.
When $p$ is indeed uniform, it is clear that $\E[Z] = 0$.
When $\dskl(p, \U) \geq \ve$, we have that
\begin{align}
\ve &\leq \sum_{(u,v) \in E} \theta_{uv}\E[X_uX_v] \label{eq:htferro1}\\
&\leq \sum_{(u,v) \in E} \tanh^{-1}(\E[X_uX_v]) \E[X_uX_v] \label{eq:htferro2}\\
&\leq \sum_{(u,v) \in E} 2\E[X_uX_v]^2 \label{eq:htferro3}\\
&\leq 2\sqrt{\frac{\ve}{n}} \sum_{(u,v) \in E} \E[X_uX_v] \label{eq:htferro4}
\end{align}
(\ref{eq:htferro1}) follows by (\ref{eq:dskl}), (\ref{eq:htferro2}) is due to Lemma~\ref{lem:ferro-structural} (since the model is ferromagnetic), (\ref{eq:htferro3}) is because $\tanh^{-1}(x) \leq 2x$ for $x \leq 0.95$, and (\ref{eq:htferro4}) is since after Step~\ref{st:filter}, we know that $\E[X_uX_v] \leq \sqrt{\ve/n}$.
This implies that $\E[Z] \geq \sqrt{\ve n/4}$.

At this point, we have that $\E[Z] = 0$ when $p$ is uniform, and $\E[Z] \geq \sqrt{\ve n/4}$ when $\dskl(p,\U) \geq \ve$.
Since the standard deviation of $Z$ is $O\left(n/\sqrt{k}\right)$, by Chebyshev's inequality, choosing $k = \Omega(n/\ve)$ suffices to distinguish the two cases with probability $\geq 9/10$.
\end{prevproof}

\section{Bounding the Variance of Functions of the Ising Model in the High-Temperature Regime} \label{sec:variance}
In this section, we describe how we can bound the variance of our statistics on the Ising model in high temperature.
Due to the complex structure of dependences, it can be challenging to obtain non-trivial bounds on the variance of even relatively simple statistics.
In particular, to apply our learn-then-test framework of Section \ref{sec:learn-and-test}, we must bound the variance of statistics of the form $Z' = \sum_{u\neq v} c_{uv}X_uX_v$ (under no external field, see (\ref{eq:statistic no external field})) and $Z'_{cen} = \sum_{u \neq v} c_{uv}\left(X_u^{(1)}-X_u^{(2)} \right) \left( X_v^{(1)}-X_v^{(2)}\right)$ (under an external field, see (\ref{eq:statistic external field})).
While the variance for both the statistics is easily seen to be $O(n^2)$ if the graph has no edges, to prove variance bounds better than the trivial $O(n^4)$ for general graphs requires some work. We show the following two theorems in this section.

The first result, Theorem \ref{thm:variance bound high temperature no external field yuval}, bounds the variance of functions of the form $\sum_{u \neq v} c_{uv} X_uX_v$ under no external field which captures the statistic used for testing independence and identity by the learn-then-test framework of Section \ref{sec:learn-and-test} in the absence of an external field.

\begin{restatable}[High Temperature Variance Bound, No External Field]{theorem}{vbnofieldyuval}
\label{thm:variance bound high temperature no external field yuval}
Let $c \in [-1,1]^{V \choose 2}$ and define $f_c: \{\pm 1\}^V \rightarrow \reals$ as follows: $f_c(x)=\sum_{i \neq j} c_{\{i,j\}} x_i x_j$. Let also $X$ be distributed according to an Ising model, without node potentials (i.e. $\theta_v=0$, for all $v$), in the high temperature regime of Definition~\ref{def:dobrushin-new}. Then
$\Var \left( f_c(X)\right) 
=O({n}^{2})$.
\end{restatable}

The second result of this section, Theorem \ref{thm:variance bound high temperature with external field yuval}, bounds the variance of functions of the form $\sum_{u \neq v} c_{uv} (X_u^{(1)} - X_u^{(2)})(X_v^{(1)} - X_v^{(2)})$ which captures the statistic of interest for independence testing using the learn-then-test framework of Section \ref{sec:learn-and-test} under an external field. 
Intuitively, this modification is required to ``recenter'' the random variables. 
Here, we view the two samples from Ising model $p$ over graph $G=(V,E)$ as coming from a single Ising model $p^{\otimes 2}$ over a graph $G^{(1)} \cup G^{(2)}$ where $G^{(1)}$ and $G^{(2)}$ are identical copies of $G$.

\begin{restatable}[High Temperature Variance Bound, Arbitrary External Field]{theorem}{vbfieldyuval}
\label{thm:variance bound high temperature with external field yuval}
Let $c \in [-1,1]^{V \choose 2}$ and let $X$ be distributed according to Ising model $p^{\otimes 2}$ over graph $G^{(1)} \cup G^{(2)}$ in the high temperature regime of Definition~\ref{def:dobrushin-new} and define $g_c: \{ \pm 1\}^{V\cup V'} \rightarrow \mathbb{R}$ as follows: $g_c(x) = \sum_{\substack{u,v \in V \\ \text{s.t. } u \neq v}} c_{uv} (x_{u^{(1)}} - x_{u^{(2)}})(x_{v^{(1)}} - x_{v^{(2)}})$. Then
$\Var(g_c(X)) = O\left(n^{2}\right).$
\end{restatable}

\subsection{Technical Overview}
We will use tools from Chapter 13 of \cite{LevinPW09} to obtain the variance bounds of this section. The statistics we use in our testing algorithms are degree-2 polynomials of the Ising model. 
To begin with, we can bound the variance of linear (or degree-1) polynomials to be $O(n)$ (standard deviation $O(\sqrt{n})$) via the Lipschitz concentration lemma (Lemma~\ref{lem:lipschitz-lemma}). However, using this lemma for degree-2 yields a large variance bound of $O(n^3)$. An informal explanation for why this is the case is as follows. The bounded differences constants required to employ Lemma \ref{lem:lipschitz-lemma} for degree-2 polynomials are bounds on degree-1 polynomials of the Ising model which can be $O(n)$ in the worst case. However, since we know that the standard deviation of degree-1 polynomials is $O(\sqrt{n})$ from the above discussion, we can leverage this to improve the variance bound for degree-2 polynomials from $O(n^3)$ to $O(n^2)$.

\noindent At a high level the technique to bound the variance of a function $f$ on a distribution $\mu$ involves first defining a reversible Markov chain with $\mu$ as its stationary distribution. By studying the mixing time properties (via the spectral gap) of this Markov chain along with the second moment of the variation of $f$ when a single step is taken under this Markov chain we obtain bounds on the second moment of $f$ which consequently yield the desired variance bounds.
The Markov chain in consideration here will be the Glauber dynamics, the canonical Markov chain for sampling from an Ising model.
Glauber dynamics define a reversible, ergodic Markov chain whose stationary distribution is identical to the corresponding Ising model. 
In many relevant settings (including the high-temperature regime), the dynamics are fast mixing (i.e., they mix in time $O(n\log n)$) and hence offer an efficient way to sample from Ising models.
As stated in Section \ref{sec:preliminaries}, the Glauber dynamics are reversible and ergodic for Ising models.

Let $M$ be the reversible transition matrix for the Glauber dynamics on some Ising model $p$ and let $\g_*$ be the absolute spectral gap of $M$. The first step is to obtain a lower bound on $\g_*$. To do so, we will use the following Lemma.
\begin{lemma}
	\label{lem:hamm-contraction}
	Let $x,y \in \Omega = \{ \pm 1 \}^n$ such that $\d_H(x,y) = 1$. Let $X,Y$ denote the states obtained by executing one step of the Glauber dynamics in a high-temperature Ising model, starting at $x$ and $y$ respectively. Then, there exists a coupling of $X,Y$ such that
	$$\E\left[ \d_H(X,Y) \right] \le \left(1 - \frac{\eta}{n}\right).$$
\end{lemma}
Lemma \ref{lem:hamm-contraction} follows almost directly from the proof of Theorem 15.1 in \cite{LevinPW09} after a simple generalization to the Ising model as we consider it here (non-uniform edge parameters).

\begin{claim}
\label{clm:spectral_gap}
In the high-temperature regime/under Dobrushin conditions, $\g_* \geq \Omega\left(\frac{1}{n}\right)$ under an arbitrary external field.
\end{claim}
\begin{proof}
	We will use Theorem 13.1 of~\cite{LevinPW09} applied on the metric space $\Omega = \{ \pm 1 \}^n$ with the metric being the Hamming distance, i.e. $\rho(x,y) = \d_H(x,y)$. To obtain a bound on the absolute spectral gap of $M$ by this Theorem, we would first need to bound, for every $x,y \in \Omega$, the expected contraction in Hamming distance under the best coupling of two executions of one step of the Glauber dynamics, one starting from state $x$ and the other from $y$. Employing the path coupling theorem of Bubley and Dyer (Theorem 14.6 of \cite{LevinPW09}), we only need to show contraction for states $x$ and $y$ which are adjacent in the Glauber dynamics chain, i.e. $x,y$ for which $\d_H(x,y)=1$ which we have from Lemma \ref{lem:hamm-contraction}. Hence we get that $\g_* \ge \frac{\eta}{n}$ implying the claim.
\end{proof}

For a function $f$, define the Dirichlet form as 
$$\mathcal{E}(f) = \frac12\sum_{x,y \in \{\pm 1\}^n} [f(x) - f(y)]^2 \p(x)M(x,y).$$
This can be interpreted as the expected square of the difference in the function, when a step is taken at stationarity. That is,
\begin{equation}
\mathcal{E}(f) = \frac12 \E_{\substack{X \sim p, \\ Y \sim M(X,\cdot)}}\left[(f(X) - f(Y))^2\right] \label{eq:dirichlet}
\end{equation}
where $x$ is drawn from the Ising distribution and $y$ is obtained by taking a step in the Glauber dynamics starting from $x$.
We now state a slight variant of Remark 13.13 of~\cite{LevinPW09}, which we will use as Lemma~\ref{lem:dirichlet_variance_bound}.
\begin{lemma}
\label{lem:dirichlet_variance_bound}
For a reversible transition matrix $P$ on state space $\Omega$ with stationary distribution $\pi$, let 
$$\mathcal{E}(f) := \frac{1}{2}\sum_{x,y \in \Omega} (f(x) - f(y))^2\pi(x)P(x,y),$$
where $f$ is a function on $\Omega$ such that $\Var_{\pi}(f) > 0$. Also let $\g_*$ be the absolute spectral gap of $P$.
Then
$$\g_* \le \frac{\mathcal{E}(f)}{\Var_{\pi}(f)}.$$
\end{lemma}
\begin{note}
	Remark 13.13 in \cite{LevinPW09} states a bound on the spectral gap as opposed to the absolute spectral gap bound which we use here. However, the proof of Remark 13.13 also works for obtaining a bound on the absolute spectral gap $\g_*$.
\end{note}

\subsection{Bounding Variance of $Z'$ Under No External Field}
We prove Theorem \ref{thm:variance bound high temperature no external field yuval} now. We first begin by bounding the variance of linear functions of the form $l_c(x) = \sum_u c_u x_u$. 
\begin{lemma}
	\label{lem:variance-linear}
	Let $X \sim p$ where $p$ is an Ising model on $G=(V,E)$ with no external field. Then 
	$\Var\left[ l_c(X) \right] \le 16 \sum_u c_u^2 $.
\end{lemma}
\begin{proof}
	The proof follows directly from Lemma \ref{lem:lipschitz-lemma}. First note that $\E[l_c(X)] = 0$.
	\begin{align*}
	\Var\left[ l_c(X) \right] = \E\left[ l_c(X)^2 \right] = \int_{0}^{\infty} \Pr\left[l_c(X)^2 \ge t\right] dt \le \int_{0}^{\infty} 2\exp\left(-\frac{t}{8\sum_u c_u^2} \right)dt \le 16\sum_u c_u^2,
	\end{align*}
	where we used that $\E[X] = \int_0^{\infty} \Pr[X \ge t]dt$ for a non-negative random variable $X$ and then applied Lemma \ref{lem:lipschitz-lemma} followed by simple calculations.
\end{proof}

Now consider the function $f_c(x) = \sum_{u, v} c_{uv} x_u x_v$ where $c \in [-1,1]^{|V| \choose 2}$.

\begin{claim}
\label{clm:dirichlet_nofield}
For an Ising model under no external field, $\mathcal{E}(f_c) =  O(n)$.
\end{claim}
\begin{proof}
	Recall the definition of $\mathcal{E}(f_c)$. Suppose node $u$ was selected by the Glauber dynamics when updating $X$ to $Y$. In this case, denote $f_c(X) - f_c(Y) = \nabla_u f_c(X)$.  $\nabla_u f_c(X)$ is either equal to $\sum_v c_{uv}X_v$  or 0 depending on whether node $u$ flipped sign or not during the Glauber update. Denote, by $E_u$, the event that node $u$ is selected by Glauber dynamics when updating $X$.
	First, we bound $\E[\nabla_u f_c(X)^2 | E_u]$.
	\begin{align}
	\E\left[(\nabla_u f_c(X))^2 | E_u \right] &\le \E\left[ \left(\sum_v c_{uv}X_v\right)^2 | E_u \right] = \E\left[ \left(\sum_v c_{uv}X_v\right)^2  \right] \label{eq:diri1}\\
	&= \Var\left( \sum_v c_{uv}X_v \right) \le  16n, \label{eq:diri2}
	\end{align}
	where (\ref{eq:diri1}) follows because $(\nabla_u f_c(X))^2 \le  \left(\sum_v c_{uv}X_v\right)^2$ with probability 1, and in (\ref{eq:diri2}) we used Lemma \ref{lem:variance-linear}.
Now,
\begin{align}
\E\left[(f_c(X) - f_c(Y))^2 \right] &= \sum_{u \in V} \frac{1}{n}\E\left[(f_c(X) - f_c(Y))^2 |  E_u\right] \label{eq:diri3}\\
&= \sum_{u \in V} \frac{1}{n}\E\left[( \nabla_u f_c(X))^2 |  E_u\right] \le \sum_{u \in V} 16 \le 16n \label{eq:diri4},
\end{align}
where (\ref{eq:diri3}) follows because $\Pr[E_u] = 1/n$ under the single-site Glauber dynamics we consider and (\ref{eq:diri4}) follows from (\ref{eq:diri2}).
\end{proof}

Claim \ref{clm:spectral_gap} together with Claim \ref{clm:dirichlet_nofield} are sufficient to conclude an upper bound on the variance of $f_c$, by using Lemma \ref{lem:dirichlet_variance_bound}, thus giving us Theorem \ref{thm:variance bound high temperature no external field yuval}.

\subsection{Bounding Variance of $Z'_{cen}$ Under Arbitrary External Field}
\label{sec:variance bound external field yuval}
Under the presence of an external field, we saw that we need to appropriately center our statistics to achieve low variance. The function $g_c(x)$ of interest now is defined over the 2-sample Ising model $p^{\otimes 2}$ and is of the form
$$g_{c}(x) = \sum_{u,v} c_{uv} \left(x_u^{(1)} - x_u^{(2)}\right) \left(x_v^{(1)} - x_v^{(2)}\right)$$
where now $x,y \in \{ \pm 1\}^{2|V|}$.
First, note that the 2-sample Ising model also satisfies Dobrushin's condition and hence the absolute spectral gap for $p^{\otimes 2}$ is also at least $ \Omega(1/n)$ in this setting.
Now we bound $\mathcal{E}(g_{c})$.

We first begin by bounding the variance of linear functions of the form $ll_c(x) = \sum_u c_u \left(x_u^{(1)} - x_u^{(2)} \right) $. 
The proof is identical to that of Lemma~\ref{lem:variance-linear}.
\begin{lemma}
	\label{lem:variance-linear-ext}
	Let $X \sim p$ where $p$ is an Ising model on $G=(V,E)$ with arbitrary external field. Then,
	$$\Var\left[ ll_c(X) \right] \le 16 \sum_u c_u^2. $$
\end{lemma}

This allows us to bound the Dirichlet form of bilinear statistics.
The proof is identical to the proof of Claim~\ref{clm:dirichlet_nofield}.
\begin{claim}
\label{clm:dirichlet_arbitraryfield}
For an Ising model under an arbitrary external field, $\mathcal{E}(g_{c}) = O(n)$.
\end{claim}

Similar to before, Claim \ref{clm:spectral_gap} together with Claim \ref{clm:dirichlet_arbitraryfield} are sufficient to conclude an upper bound on the variance of $g_c$, by using Lemma \ref{lem:dirichlet_variance_bound}, thus giving us Theorem~\ref{thm:variance bound high temperature with external field yuval}.

\section{Lower Bounds}
\label{sec:lb}
\subsection{Dependences on $n$}
Our first lower bounds show dependences on $n$, the number of nodes, in the complexity of testing Ising models.

To start, we prove that uniformity testing on product measures over a binary alphabet requires $\Omega(\sqrt{n}/\ve)$ samples.
Note that a binary product measure corresponds to the case of an Ising model with no edges.
This implies the same lower bound for identity testing, but \emph{not} independence testing, as a product measure always has independent marginals, so the answer is trivial.
\begin{theorem}
\label{thm:product-lb}
There exists a constant $c > 0$ such that any algorithm, given sample access to an Ising model $p$ with no edges (i.e., a product measure over a binary alphabet), which distinguishes between the cases $p = \U$ and $\dskl(p, \U) \geq \ve$ with probability at least $99/100$ requires $k \geq c\sqrt{n}/\ve$ samples.
\end{theorem}

Next, we show that any algorithm which tests uniformity of an Ising model requires $\Omega(n/\ve)$ samples.
In this case, it implies the same lower bounds for independence and identity testing.
\begin{theorem}
\label{thm:linear-lb}
There exists a constant $c > 0$ such that any algorithm, given sample access to an Ising model $p$, which distinguishes between the cases $p = \U$ and $\dskl(p, \U) \geq \ve$ with probability at least $99/100$ requires $k \geq cn/\ve$ samples.
This remains the case even if $p$ is known to have a tree structure and only ferromagnetic edges.
\end{theorem}

The lower bounds use Le Cam's two point method which constructs a family of distributions $\mathcal{P}$ such that the distance between any $P \in \mathcal{P}$ and a particular distribution $Q$ is large (at least $\ve$). 
But given a $P \in \mathcal{P}$ chosen uniformly at random, it is hard to distinguish between $P$ and $Q$ with at least 2/3 success probability unless we have sufficiently many samples. 

Our construction for product measures is inspired by Paninski's lower bound for uniformity testing \cite{Paninski08}.
We start with the uniform Ising model and perturb each node positively or negatively by $\sqrt{\ve/n}$, resulting in a model which is $\ve$-far in $\dskl$ from $\U$. 
The proof appears in Section \ref{sec:product-lb}.

Our construction for the linear lower bound builds upon this style of perturbation.
In the previous construction, instead of perturbing the node potentials, we could have left the node marginals to be uniform and perturbed the edges of some fixed, known matching to obtain the same lower bound.
To get a linear lower bound, we instead choose a \emph{random} perfect matching, which turns out to require quadratically more samples to test.
Interestingly, we only need ferromagnetic edges (i.e., positive perturbations), as the randomness in the choice of matching is sufficient to make the problem harder.
Our proof is significantly more complicated for this case, and it uses a careful combinatorial analysis involving graphs which are unions of two perfect matchings. 
The lower bound is described in detail in Section \ref{sec:linear-lb}.

\begin{remark}
\label{rem:lb works for tv}
Similar lower bound constructions to those of Theorems \ref{thm:product-lb} and \ref{thm:linear-lb} also yield $\Omega(\sqrt{n}/\ve^2)$ and $\Omega(n/\ve^2)$ for the corresponding testing problems when $\dskl$ is replaced with $\dtv$.
In our constructions, we describe families of distributions which are $\ve$-far in $\dskl$.
This is done by perturbing certain parameters by a magnitude of $\Theta(\sqrt{\ve/n})$.
We can instead describe families of distributions which are $\ve$-far in $\dtv$ by performing perturbations of $\Theta(\ve/\sqrt{n})$, and the rest of the proofs follow similarly.
\end{remark}

\subsection{Dependences on $h, \beta$}
Finally, we show that dependences on the $h$ and $\b$ parameters are, in general, necessary for independence and identity testing.

\begin{theorem}
\label{thm:bh-lb}
There is a linear lower bound on the parameters $h$ and $\b$ for testing problems on Ising models.
More specifically,
\begin{itemize}
\item There exists a constant $c>0$ such that, for all $\ve < 1$ and $\b \geq 0$, any algorithm, given sample access to an Ising model $p$, which distinguishes between the cases $p \in \I$ and $\dskl(p,\I) \geq \ve$ with probability at least $99/100$ requires $k \geq c\b/\ve$ samples.
\item There exists constants $c_1, c_2>0$ such that, for all $\ve < 1$ and $\b \geq c_1 \log (1/\ve)$, any algorithm, given a description of an Ising model $q$ with no external field (i.e., $h = 0$) and has sample access to an Ising model $p$, and which distinguishes between the cases $p = q$ and $\dskl(p,q) \geq \ve$ with probability at least $99/100$ requires $k \geq c_2\b/\ve$ samples.
\item There exists constants $c_1, c_2>0$ such that, for all $\ve < 1$ and $h \geq c_1 \log (1/\ve)$, any algorithm, given a description of an Ising model $q$ with no edge potentials(i.e., $\b = 0$) and has sample access to an Ising model $p$, and which distinguishes between the cases $p = q$ and $\dskl(p,q) \geq \ve$ with probability at least $99/100$ requires $k \geq c_2 h/\ve$ samples.
\end{itemize}
\end{theorem}
The construction and analysis appear in Section \ref{sec:bh-lb}.
This lower bound shows that the dependence on $\b$ parameters by our algorithms cannot be avoided in general, though it may be sidestepped in certain cases.
Notably, we show that testing independence of a forest-structured Ising model under no external field can be done using $\tilde{O}\left(\frac{n}{\ve} \right)$ samples (Theorem \ref{thm:forests-independence}). 

\subsection{Lower Bound Proofs} 
\label{sec:lbAppendix}
\subsubsection{Proof of Theorem \ref{thm:product-lb}}
\label{sec:product-lb}
This proof will follow via an application of Le Cam's two-point method.
More specifically, we will consider two classes of distributions $\mathcal{P}$ and $\mathcal{Q}$ such that:
\begin{enumerate}
\item $\mathcal{P}$ consists of a single distribution $p \triangleq \U$;
\item $\mathcal{Q}$ consists of a family of distributions such that for all distributions $q \in \mathcal{Q}$, $\dskl(p,q) \geq \ve$;
\item There exists some constant $c >0$ such that any algorithm which distinguishes $p$ from a uniformly random distribution $q \in \mathcal{Q}$ with probability $\geq 2/3$ requires $\geq c\sqrt{n}/\ve$ samples.
\end{enumerate}
The third point will be proven by showing that, with $k < c\sqrt{n}/\ve$ samples, the following two processes have small total variation distance, and thus no algorithm can distinguish them:
\begin{itemize}
\item The process $p^{\otimes k}$, which draws $k$ samples from $p$;
\item The process $\bar q^{\otimes k}$, which selects $q$ from $\mathcal{Q}$ uniformly at random, and then draws $k$ samples from $q$.
\end{itemize}
We will let $p^{\otimes k}_i$ be the process $p^{\otimes k}$ restricted to the $i$th coordinate of the random vectors sampled, and $\bar q^{\otimes k}_i$ is defined similarly.

We proceed with a description of our construction.
Let $\d = \sqrt{3\ve/2n}$.
As mentioned before, $\mathcal{P}$ consists of the single distribution $p \triangleq \U$, the Ising model on $n$ nodes with $0$ potentials on every node and edge.
Let $\mathcal{M}$ be the set of all $2^n$ vectors in the set $\{\pm \d\}^n$. 
For each $M = (M_1, \dots, M_n) \in \mathcal{M}$, we define a corresponding $q_M \in \mathcal{Q}$ where the node potential $M_i$ is placed on node $i$.

\begin{proposition}
\label{prop:product-far}
For each $q \in \mathcal{Q}$, $\dskl(q, \U) \geq \ve$.
\end{proposition}
\begin{proof}
Recall that $\dskl(q, \U) = \sum_{v \in V} \d \tanh(\d).$
Note that $\tanh(\d) \geq 2\d/3$ for all $\d \leq 1$, which can be shown using a Taylor expansion.
Therefore $\dskl(q, \U) \geq n\cdot \d \cdot 2\d/3 = 2n\d^2/3 = \ve.$
\end{proof}

The goal is to upper bound $\dtv(p^{\otimes k}, \bar q^{\otimes k})$.
Our approach will start with manipulations similar to the following lemma from \cite{AcharyaD15}, which follows by Pinsker's and Jensen's inequalities.
\begin{lemma}
\label{lem:pinsker-jensen}
For any two distributions $p$ and $q$, $2\dtv^2(p,q) \leq \dkl(q, p) \leq \log \E_{q} \left[ \frac{q}{p}\right].$
\end{lemma}

Similarly:
$2\dtv^2(p^{\otimes k},\bar q^{\otimes k}) \leq \dkl(\bar q^{\otimes k},p^{\otimes k}) = n\dkl(\bar q_1^{\otimes k}, p_1^{\otimes k}) \leq  n \log \E_{\bar q^{\otimes k}_1}\left[\frac{\bar q^{\otimes k}_1}{p^{\otimes k}_1}\right].$

We proceed to bound the right-hand side.
To simplify notation, let $p_+ = e^{\d}/(e^{\d} + e^{-\d})$ be the probability that a node with parameter $\d$ takes the value $1$.

\begin{eqnarray*}
&&\E_{\bar{q}_1^{\otimes k}}\left[\frac{\bar{q}_{1}^{\otimes k}}{p_1^{\otimes k}}\right] = \sum_{k_1=0}^{k} \frac{(\bar{q}_{1}^{\otimes k}(k_1))^2}{p_{1}^{\otimes k}(k_1)}
= \sum_{k_1 = 0}^{k} \frac{\left(\frac{1}{2}\binom{k}{k_1}(p_+)^{k_1}(1-p_+)^{k-k_1} + \frac{1}{2}\binom{k}{k-k_1}(p_+)^{k - k_1}(1-p_+)^{k_1}\right)^2}{\binom{k}{k_1}(1/2)^k}\\
&=& \frac{2^k}{4}\sum_{k_1 = 0}^{k} \binom{k}{k_1} \left( (p_+)^{2k_1}(1-p_+)^{2(k-k_1)} + (p_+)^{2(k-k_1)}(1-p_+)^{2k_1} + 2(p_+(1-p_+))^k \right)\\
&=& \frac{2^k}{2}(p_+(1-p_+))^k\sum_{k_1=0}^{k} \binom{k}{k_1} + 2\cdot\frac{2^k}{4} \sum_{k_1=0}^{k} \left(\binom{k}{k_1}(p_+^2)^{k_1}((1-p_+)^2)^{k-k_1}\right)
\end{eqnarray*}
Using the Binomial theorem, the value for $p_+$, and hyperbolic trigenometric identities:
\begin{align*}
\E_{\bar{q}_{1}^{\otimes k}}\left[\frac{\bar{q}_{1}^{\otimes k}}{p_{1}^{\otimes k}}\right] &= \frac{4^k}{2}(p_+(1-p_+))^k + \frac{2^k}{2}\left( 2p_+^2 + 1-2p_+ \right)^k \\
&= \frac{1}{2}\left(\left(\sech^2(\d)\right)^k + \left(1 + \tanh^2(\d)\right)^k \right) 
\leq 1 + \binom{k}{2}\d^4 
= 1 + \binom{k}{2}\frac{9\ve^2}{4n^2}.
\end{align*}

This gives us that
$$2\dtv^2(p^{\otimes k},\bar q^{\otimes k})\leq n \log \left( 1 + \binom{k}{2}\frac{9\ve^2}{4n^2}\right) \leq \frac{9k^2\ve^2}{4n}.$$
If $k < 0.9 \cdot \sqrt{n}/\ve$, then $\dtv^2(p^{\otimes k},\bar q^{\otimes k}) < 49/50$,
completing the proof of Theorem \ref{thm:product-lb}.

\subsubsection{Proof of Theorem \ref{thm:linear-lb}}
\label{sec:linear-lb}
This lower bound similarly applies Le Cam's two-point method, as described in the previous section.
We proceed with a description of our construction.
Assume that $n$ is even.
As before, $\mathcal{P}$ consists of the single distribution $p \triangleq \U$, the Ising model on $n$ nodes with $0$ potentials on every node and edge.
Let $\mathcal{M}$ denote the set of all $(n-1)!!$ perfect matchings on the clique on $n$ nodes.
Each $M \in \mathcal{M}$ defines a corresponding $q_M \in \mathcal{Q}$, where the potential $\d = \sqrt{3\ve/n}$ is placed on each edge present in the graph.

The following proposition follows similarly to Proposition \ref{prop:product-far}.
\begin{proposition}
For each $q \in \mathcal{Q}$, $\dskl(q, \U) \geq \ve$.
\end{proposition}

The goal is to upper bound $\dtv(p^{\otimes k}, \bar q^{\otimes k})$.
We apply Lemma \ref{lem:pinsker-jensen} to $2\dtv^2(p^{\otimes k}, \bar q^{\otimes k})$ and focus on the quantity inside the logarithm.
Let $X^{(i)} \in \{\pm 1\}^n$ represent the realization of the $i$th sample and $X_u \in \{\pm 1\}^k$ represent the realization of the $k$ samples on node $u$.
Let $H(.,.)$ represent the Hamming distance between two vectors, and for sets $S_1$ and $S_2$, let $S = S_1 \uplus S_2$ be the multiset addition operation (i.e., combine all the elements from $S_1$ and $S_2$, keeping duplicates).
Let $M_0$ be the perfect matching with edges $(2i-1,2i)$ for all $i \in [n/2]$.
\[
\E_{\bar q^{\otimes k}} \left[ \frac{\bar q^{\otimes k}}{p^{\otimes k}}\right] = \sum_{X = (X^{(1)}, \dots, X^{(k)})} \frac{(\bar q^{\otimes k}(X))^2}{p^{\otimes k}(X)} 
=2^{nk}  \sum_{X = (X^{(1)}, \dots, X^{(k)})} (\bar q^{\otimes k}(X))^2
\]
We can expand the inner probability as follows.
Given a randomly selected perfect matching, we can break the probability of a realization $X$ into a product over the edges.
By examining the PMF of the Ising model, if the two endpoints of a given edge agree, the probability is multiplied by a factor of $\left(\frac{e^{\d}}{2(e^{\d}+e^{-\d})}\right)$, and if they disagree, a factor of $\left(\frac{e^{-\d}}{2(e^{\d}+e^{-\d})}\right)$.
Since (given a matching) the samples are independent, we take the product of this over all $k$ samples.
We average this quantity using a uniformly random choice of perfect matching.
Mathematically:
\begin{align}
&2^{nk}  \sum_{X = (X^{(1)}, \dots, X^{(k)})} \left(\frac{1}{(n-1)!!} \sum_{M \in \mathcal{M}} \prod_{(u,v) \in M} \prod_{i=1}^k \left(\frac{e^{\d}}{2(e^{\d}+e^{-\d})}\right)^{\mathbbm{1}(X^{(i)}_u = X^{(i)}_v)}\left(\frac{e^{-\d}}{2(e^{\d}+e^{-\d})}\right)^{\mathbbm{1}(X^{(i)}_u \neq X^{(i)}_v)}\right)^2\nonumber\\
&= \left(\frac{e^\d}{e^\d + e^{-\d}}\right)^{nk} \frac{1}{(n-1)!!^2}\sum_{X = (X^{(1)}, \dots, X^{(k)})} \sum_{M_1, M_2 \in \mathcal{M}} \prod_{(u,v) \in M_1 \uplus M_2}  \exp(-2\d H(X_u,X_v))\nonumber
\end{align}

By symmetry, we can fix $M_1$ to be $M_0$ and multiply by a factor of $(n-1)!!$:


\begin{align}
&\left(\frac{e^\d}{e^\d + e^{-\d}}\right)^{nk} \frac{1}{(n-1)!!}\sum_{M \in \mathcal{M}} \sum_{X = (X^{(1)}, \dots, X^{(k)})} \prod_{(u,v) \in M_0 \uplus M}  \exp(-2\d H(X_u,X_v)) \nonumber \\
&= \left(\frac{e^\d}{e^\d + e^{-\d}}\right)^{nk} \frac{1}{(n-1)!!}\sum_{M \in \mathcal{M}} \left(\sum_{X^{(1)}} \prod_{(u,v) \in M_0 \uplus M}  \exp\left(-2\d H\left(X^{(1)}_u,X^{(1)}_v\right)\right)\right)^k\nonumber
\end{align}

We observe that multiset union of two perfect matchings will form a collection of even length cycles (if they contain the same edge, this forms a 2-cycle), and this can be rewritten as follows.

\begin{align}
&\left(\frac{e^\d}{e^\d + e^{-\d}}\right)^{nk}\frac{1}{(n-1)!!}\sum_{M \in \mathcal{M}} \left(\sum_{X^{(1)}} \prod_{\substack{\text{cycles} C \\  \in M_0 \uplus M}} \prod_{(u,v) \in C}  \exp\left(-2\d H\left(X^{(1)}_u,X^{(1)}_v\right)\right)\right)^k \nonumber \\
&= \left(\frac{e^\d}{e^\d + e^{-\d}}\right)^{nk} \frac{1}{(n-1)!!}\sum_{M \in \mathcal{M}} \left(\prod_{\substack{\text{cycles } C \\  \in M_0 \uplus M}} \sum_{X^{(1)}_C} \prod_{(u,v) \in C}  \exp\left(-2\d H\left(X^{(1)}_u,X^{(1)}_v\right)\right)\right)^k \label{eq:prev-counting-arg}
\end{align}

We now simplify this using a counting argument over the possible realizations of $X^{(1)}$ when restricted to edges in cycle $C$.
Start by noting that
$$\sum_{X^{(1)}_C} \prod_{(u,v) \in C} (e^{2\d})^{ -H\left(X^{(1)}_u,X^{(1)}_v\right)} = 2 \sum_{i=0}^{n/2} \binom{|C|}{2i}(e^{2\d})^{-2i}.$$
This follows by counting the number of possible ways to achieve a particular Hamming distance over the cycle, observing that only even values are allowed by a parity argument.
This is twice the sum over the even terms in the binomial expansion of $(1 + e^{-2\d})^{|C|}$.
The odd terms may be eliminated by adding $(1 - e^{-2\d})^{|C|}$, and thus (\ref{eq:prev-counting-arg}) is equal to the following:

\begin{align}
&\left(\frac{e^\d}{e^\d + e^{-\d}}\right)^{nk} \frac{1}{(n-1)!!}\sum_{M \in \mathcal{M}} \left(\prod_{\substack{\text{cycles } C \\  \in M_0 \uplus M}} (1 + e^{-2\d})^{|C|} + (1 - e^{-2\d})^{|C|}\right)^k\nonumber \\
&= \left(\frac{e^\d}{e^\d + e^{-\d}}\right)^{nk} \frac{1}{(n-1)!!}\sum_{M \in \mathcal{M}} \left(\prod_{\substack{\text{cycles } C \\  \in M_0 \uplus M}} \left(\frac{e^{\d}+e^{-\d}}{e^{\d}}\right)^{|C|} \left(1 + \left(\frac{e^{\d} - e^{-\d}}{e^{\d} + e^{-\d}}\right)^{|C|}\right)\right)^k \nonumber \\
&= \E\left[\left(\prod_{\substack{\text{cycles } C \\  \in M_0 \uplus M}} \left(1 + \tanh^{|C|}(\d)\right)\right)^k\right] \leq
\E\left[\left(\prod_{\substack{\text{cycles } C \\  \in M_0 \uplus M}} \exp\left(\d^{|C|}\right)\right)^k\right]. 
 \label{eq:overlap-expectation}
\end{align}
where the expectation is from choosing a uniformly random perfect matching $M \in \mathcal{M}$.
At this point, it remains only to bound Equation (\ref{eq:overlap-expectation}).
For our purposes, it turns out that the $2$-cycles will be the dominating factor, and we use the following crude upper bound.
Let $\z$ be a random variable representing the number of $2$-cycles in $M_0 \uplus M$, i.e., the number of edges shared by both perfect matchings.
$$
\E\left[\left(\prod_{\substack{\text{cycles } C \\  \in M_0 \uplus M}} \exp\left(\d^{|C|}\right)\right)^k\right]
= \E\left[\left(\prod_{\substack{\text{cycles } C \\  \in M_0 \uplus M \\ |C| \geq 4}} \exp\left(\d^{|C|}\right)\right)^k\exp\left(\d^2\z k\right)\right  ]
\leq \exp\left(\d^4 \cdot n/4 \cdot k \right) \E\left[\exp\left(\d^2 \z k\right)\right], $$
where in the last inequality, we used the facts that $\d^{|C|}$ is maximized for $|C| \geq 4$ when $|C| = 4$, and that there are at most $n/4$ cycles of length at least $4$.

We examine the distribution of $\z$.
Note that $\E[\z] = \frac{n}{2} \cdot \frac{1}{n-1} = \frac{n}{2(n-1)}.$
More generally, for any integer $z \leq n/2$,
$\E[\z - (z-1)| \z \geq z-1] = \frac{n-2z + 2}{2} \cdot \frac{1}{n-2z+1} = \frac{n-2z+2}{2(n-2z+1)}.$
By Markov's inequality, $\Pr[\z \geq z | \z \geq z - 1 ] = \Pr[\z - (z-1) \geq 1 | \z \geq z - 1] \leq \frac{n-2z + 2}{2(n-2z + 1)}.$
Therefore,$\Pr[\z \geq z] = \prod_{i=1}^{z} \Pr[\z \geq i| \z \geq i - 1] \leq \prod_{i=1}^z \frac{n-2i + 2}{2(n-2i+1)}.$
In particular, note that for all $z < n/2$, 
$\Pr[\z \geq z] \leq (2/3)^z.$

We return to considering the expectation above:
\begin{align*}
\E\left[\exp\left(\d^2 \z k\right)\right] 
&= \sum_{z = 0}^{n/2} \Pr[\z = z] \exp\left(\d^2 z k\right) 
\leq \sum_{z = 0}^{n/2} \Pr[\z \geq z] \exp\left(\d^2 z k \right) 
\leq \frac{3}{2} \sum_{z = 0}^{n/2} (2/3)^z \exp\left(\d^2 z k \right) \\
&= \frac{3}{2} \sum_{z = 0}^{n/2} \exp\left((\d^2  k - \log (3/2))z\right) 
\leq \frac{3}{2} \cdot \frac{1}{1 - \exp\left(\d^2 k - \log(3/2)\right)},
\end{align*}
where the last inequality requires that $\exp\left(\d^2 k - \log(3/2)\right) < 1$.
This is true as long as $k < \log(3/2)/\d^2 = \frac{\log(3/2)}{3} \cdot \frac{n}{\ve}$.

Combining Lemma \ref{lem:pinsker-jensen} with the above derivation, we have that
\begin{align*}
&2\dtv^2(p^{\otimes k}, \bar q^{\otimes k}) 
\leq \log \left(\exp(\d^4 nk/4) \cdot \frac{3}{2(1 - \exp\left(\d^2 k - \log (3/2)\right))}\right) \\
&= \d^4 nk/4 + \log \left(\frac{3}{2(1 - \exp\left(\d^2 k - \log (3/2)\right))}\right) 
= \frac{9\ve^2}{4n}k + \log \left(\frac{3}{2(1 - \exp\left(3k\ve/n - \log (3/2)\right))}\right).
\end{align*}
If $k < \frac{1}{25}\cdot \frac{n}{\ve}$, then $\dtv(p^{\otimes k}, \bar q^{\otimes k}) < 49/50$,
completing the proof of Theorem \ref{thm:linear-lb}.

\subsubsection{Proof of Theorem \ref{thm:bh-lb}}
\label{sec:bh-lb}

\paragraph{Independence Testing Dependence on $\beta$:}

Consider the following two models, which share some parameter $\tau > 0$:
\begin{enumerate}
\item An Ising model $p$ on two nodes $u$ and $v$, where $\theta_u^p = \theta_v^p = \tau$ and $\theta_{uv} = 0$.
\item An Ising model $q$ on two nodes $u$ and $v$, where $\theta_u^q = \theta_v^q = \tau$ and $\theta_{uv} = \b$.
\end{enumerate}
We note that $\E[X_u^pX_v^p] = \frac{\exp{(2\t + \b)} + \exp{(-2\t + \b)} - \exp(-\b)}{\exp{(2\t + \b)} + \exp{(-2\t + \b)} + \exp(-\b)}$ and $\E[X_u^qX_v^q] = \tanh^2(\t)$.
By (\ref{eq:dskl}), these two models have $\dskl(p,q) = \b\left(\E[X_u^pX_v^p] - \E[X_u^qX_v^q]\right)$.
For any for any fixed $\beta$ sufficiently large and $\ve > 0$ sufficiently small, $\tau$ can be chosen to make $\E[X_u^pX_v^p] - \E[X_u^qX_v^q] = \frac{\ve}{\b}$.
This is because at $\tau = 0$, this is equal to $\tanh(\b)$ and for $\tau \rightarrow \infty$, this approaches $0$, so by continuity, there must be a $\tau$ which causes the expression to equal this value.
Therefore, the SKL distance between these two models is $\ve$.
On the other hand, it is not hard to see that $\dtv(p,q) = \Theta\left(\E[X_u^pX_v^p] - \E[X_u^qX_v^q]\right) = \Theta(\ve/\b)$, and therefore, to distinguish these models, we require $\Omega(\b/\ve)$ samples.

\paragraph{Identity Testing Dependence on $\beta, h$:}

Consider the following two models, which share some parameter $\tau >0$:
\begin{enumerate}
\item An Ising model $p$ on two nodes $u$ and $v$, where $\theta_{uv}^p = \beta$.
\item An Ising model $q$ on two nodes $u$ and $v$, where $\theta_{uv}^p = \beta - \tau$.
\end{enumerate}
We note that $\E[X_u^pX_v^p] = \tanh(\beta) $ and $\E[X_u^qX_v^q] = \tanh(\beta - \tau)$.
By (\ref{eq:dskl}), these two models have $\dskl(p,q) = \t\left(\E[X_u^pX_v^p] - \E[X_u^qX_v^q]\right)$.
Observe that at $\t = \b$, $\dskl(p,q) = \b \tanh(\b)$, and at $\t = \b/2$, $\dskl(p, q) = \frac{\b}{2} (\tanh(\b) - \tanh(\b/2)) = \frac{\b}{2}(\tanh(\b/2)\sech(\b)) \leq  \b \exp(-\b) \leq \ve$, where the last inequality is based on our condition that $\b$ is sufficiently large.
By continuity, there exists some $\t \in [\b/2, \b]$ such that $\dskl(p,q) = \ve$. 
On the other hand, it is not hard to see that $\dtv(p,q) = \Theta\left(\E[X_u^pX_v^p] - \E[X_u^qX_v^q]\right) = \Theta(\ve/\b)$, and therefore, to distinguish these models, we require $\Omega(\b/\ve)$ samples.

The lower bound construction and analysis for the $h$ lower bound follow almost identically, with the model $q$ consisting of a single node with parameter $h$.

\section*{Acknowledgements}
The authors would like to thank Yuval Peres for directing them towards the reference~\cite{LevinPW09} and the tools used to prove the variance bounds in this paper, and Pedro Felzenszwalb for discussing applications of Ising models in computer vision.
They would also like to thank the anonymous reviewers, who provided numerous suggestions on how to improve the presentation of the paper, and suggested a simplification for ferromagnetic models based on Griffiths' inequality.
\bibliographystyle{alpha}
\bibliography{biblio}
\appendix
\section{Weakly Learning Rademacher Random Variables}
\label{sec:weak-learning-bern}
In this section, we examine the concept of ``weakly learning'' Rademacher random variables.
This problem we study is classical, but our regime of study and goals are slightly different.
Suppose we have $k$ samples from a random variable, promised to either be $Rademacher(1/2 + \l)$ or $Rademacher(1/2 - \l)$, for some $0 < \l \leq 1/2$.
How many samples do we need to tell which case we are in?
If we wish to be correct with probability (say) $\geq 2/3$, it is folklore that $k = \Theta(1/\l^2)$ samples are both necessary and sufficient.
In our weak learning setting, we focus on the regime where we are sample limited (say, when $\l$ is very small), and we are unable to gain a constant benefit over randomly guessing.
More precisely, we have a budget of $k$ samples from some $Rademacher(p)$ random variable, and we want to guess whether $p > 1/2$ or $p < 1/2$.
The ``margin'' $\l = |p - 1/2|$ may not be precisely known, but we still wish to obtain the maximum possible advantage over randomly guessing, which gives us probability of success equal to $1/2$. 
We show that with any $k \leq 1/4\l^2$ samples, we can obtain success probability $1/2 + \Omega(\l \sqrt{k})$.
This smoothly interpolates within the ``low sample'' regime, up to the point where $k = \Theta(1/\l^2)$ and folklore results also guarantee a constant probability of success.
We note that in this low sample regime, standard concentration bounds like Chebyshev and Chernoff give trivial guarantees, and our techniques require a more careful examination of the Binomial PMF.

We go on to examine the same problem under alternate centerings -- where we are trying to determine whether $p > \m$ or $p < \m$, generalizing the previous case where $\m = 1/2$. 
We provide a simple ``recentering'' based reduction to the previous case, showing that the same upper bound holds for all values of $\m$.
We note that our reduction holds even when the centering $\m$ is not explicitly known, and we only have limited sample access to $Rademacher(\m)$.

We start by proving the following lemma, where we wish to determine the direction of bias with respect to a zero-mean Rademacher random variable.
\begin{lemma}
\label{lem:weakLearningBernoulli}
Let $X_1, \dots, X_k$ be i.i.d.\ random variables, distributed as $Rademacher(p)$ for any $p \in [0,1]$.
There exists an algorithm which takes $X_1, \dots, X_k$ as input and outputs a value $b \in \{\pm 1\}$, with the following guarantees:
there exists constants $c_1, c_2 > 0$ such that for any $p \neq \frac12$,
$$\Pr\left(b = \sign\left(\l\right)\right) \geq 
\begin{cases}
\frac12 + c_1 |\l| \sqrt{k} &\text{if } k \leq \frac{1}{4\l^2}\\
\frac12 + c_2 &\text{otherwise,}
\end{cases}
$$
where $\l = p - \frac12$.
If $p = \frac12$, then $b \sim Rademacher\left(\frac12\right)$.
\end{lemma}
\begin{proof}
The algorithm is as follows: let $S = \sum_{i=1}^k X_i$.
If $S \neq 0$, then output $b = \sign(S)$, otherwise output $b \sim Rademacher\left(\frac12\right)$.

The $p = 1/2$ case is trivial, as the sum $S$ is symmetric about $0$.
We consider the case where $\l > 0$ (the negative case follows by symmetry) and when $k$ is even (odd $k$ can be handled similarly).
As the case where $k >\frac{1}{4\l^2}$ follows by a Chernoff bound, we focus on the former case, where $\l \leq \frac{1}{2\sqrt{k}}$.
By rescaling and shifting the variables, this is equivalent to lower bounding $\Pr\left(Binomial\left(k,\frac12 + \l\right) \geq \frac{k}{2}\right)$.
By a symmetry argument, this is equal to 
$$\frac12 + \dtv\left(Binomial\left(k, \frac12 - \l\right), Binomial\left(k, \frac12 + \l\right)\right).$$ 
It remains to show this total variation distance is $\Omega(\l\sqrt{k})$.
\begin{align}
&&\dtv\left(Binomial\left(k, \frac12 - \l\right), Binomial\left(k, \frac12 + \l\right)\right) \nonumber \\
&\geq&\dtv\left(Binomial\left(k, \frac12 \right), Binomial\left(k, \frac12 + \l\right)\right) \nonumber \\
&\geq&k \min_{\ell \in \{\lceil k/2 \rceil, \dots, \lceil k/2 + k\l\rceil\}}\int_{1/2}^{1/2 + \l}  \Pr\left(Binomial\left(k - 1, u\right) = l - 1\right) du \label{eq:wklrn1} \\
&\geq&\l k \cdot \Pr\left(Binomial\left(k - 1, 1/2 + \l\right) = k/2\right) \nonumber \\
&=&\l k \cdot \binom{k-1}{k/2} \left(\frac{1}{2} + \l\right)^{k/2}\left(\frac12 - \l\right)^{k/2 -1} \nonumber \\
&\geq& \Omega(\l k) \cdot \sqrt{\frac{1}{2k}} \left(1 + \frac{1}{\sqrt{k}}\right)^{k/2}\left(1 - \frac{1}{\sqrt{k}}\right)^{k/2} \label{eq:wklrn2}\\
&=& \Omega(\l \sqrt{k}) \cdot \left(1 - \frac{1}{k}\right)^{k/2} \nonumber \\
&\geq& \Omega(\l \sqrt{k}) \cdot \exp\left(-1/2\right)\left(1 - \frac1{k}\right)^{1/2} \label{eq:wklrn3}\\
&=& \Omega(\l \sqrt{k}), \nonumber
\end{align}
as desired.

(\ref{eq:wklrn1}) applies Proposition 2.3 of \cite{AdellJ06}.
(\ref{eq:wklrn2}) is by an application of Stirling's approximation and since $\l \leq \frac1{2\sqrt{k}}$.
(\ref{eq:wklrn3}) is by the inequality $\left(1 - \frac{c}{k}\right)^k \geq \left(1 - \frac{c}{k}\right)^c \exp(-c)$.
\end{proof}

We now develop a corollary allowing us to instead consider comparisons with respect to different centerings.
\begin{corollary}
\label{cor:weakLearningBernoulliGeneral}
Let $X_1, \dots, X_k$ be i.i.d.\ random variables, distributed as $Rademacher(p)$ for any $p \in [0,1]$.
There exists an algorithm which takes $X_1, \dots, X_k$ and $q \in [0,1]$ as input and outputs a value $b \in \{\pm 1\}$, with the following guarantees:
there exists constants $c_1, c_2 > 0$ such that for any $p \neq q$,
$$\Pr\left(b = \sign\left(\l\right)\right) \geq 
\begin{cases}
\frac12 + c_1 |\l| \sqrt{k} &\text{if } k \leq \frac{1}{4\l^2}\\
\frac12 + c_2 &\text{otherwise,}
\end{cases}
$$
where $\l = \frac{p-q}{2}$.
If $p = q$, then $b \sim Rademacher\left(\frac12\right)$.

This algorithm works even if only given $k$ i.i.d.\ samples $Y_1, \dots, Y_k \sim Rademacher(q)$, rather than the value of $q$.
\end{corollary}
\begin{proof}
Let $X \sim Rademacher(p)$ and $Y \sim Rademacher(q)$.
Consider the random variable $Z$ defined as follows.
First, sample $X$ and $Y$.
If $X \neq Y$, output $\frac{1}{2}\left(X - Y\right)$.
Otherwise, output a random variable sampled as $Rademacher\left(\frac12\right)$.
One can see that $Z \sim Rademacher\left(\frac12 + \frac{p-q}{2}\right)$.

Our algorithm can generate $k$ i.i.d.\ samples $Z_i \sim Rademacher\left(\frac12 + \frac{p-q}{2}\right)$ in this method using $X_i$'s and $Y_i$'s, where $Y_i$'s are either provided as input to the algorithm or generated according to $Rademacher(q)$.
At this point, we provide the $Z_i$'s as input to the algorithm of Lemma~\ref{lem:weakLearningBernoulli}.
By examining the guarantees of Lemma~\ref{lem:weakLearningBernoulli}, this implies the desired result.
\end{proof}

\section{Structural Result for Forest Ising Models}
\label{sec:treemarginalAppendix}

\begin{prevproof}{Lemma}{lem:trees-structural}
Consider any edge $e=(u,v) \in E$. Consider the tree $(T,E_T)$ which contains $e$. Let $n_T$ be the number of nodes in the tree. We partition the vertex set $T$ into $U$ and $V$ as follows. Remove edge $e$ from the graph and let $U$ denote all the vertices which lie in the connected component of node $u$ except $u$ itself. Similarly, let $V$ denote all the vertices which lie in the connected component of node $v$ except node $v$ itself. Hence, $T = U \cup V \cup \{u\}\cup \{v\}$. Let $X_U$ be the vector random variable which denotes the assignment of values in $\{ \pm 1\}^{|U|}$ to the nodes in $U$. $X_V$ is defined similarly. We will also denote a specific value assignment to a set of nodes $S$ by $x_S$ and $-x_S$ denotes the assignment which corresponds to multiplying each coordinate of $x_S$ by $-1$. Now we state the following claim which follows from the tree structure of the Ising model.
\begin{claim}
$\Pr\left[ X_U=x_U,X_u = 1, X_v = 1, X_V=x_V \right] = \exp(2\th_{uv})\Pr\left[ X_U=x_U,X_u = 1, X_v = -1,X_V=-x_V \right]$.
\end{claim}
In particular the above claim implies the following corollary which is obtained by marginalization of the probability to nodes $u$ and $v$.
\begin{corollary}
\label{cor:trees}
If $X$ is an Ising model on a forest graph $G=(V,E)$ with no external field, then for any edge $e=(u,v) \in E$, $\Pr\left[X_u = 1, X_v = 1 \right] = \exp(2\th_{uv})\Pr\left[ X_u = 1, X_v = -1 \right]$.
\end{corollary}

Now,
\begin{align}
& \E\left[X_uX_v \right] = \Pr\left[X_uX_v = 1\right] - \Pr\left[X_uX_v = -1 \right]\\
& ~~= 2Pr\left[X_u=1,X_v = 1\right] - 2\Pr\left[X_u=1, X_v = -1 \right] \label{eq:sym} \\
& ~~= \frac{2Pr\left[X_u=1,X_v = 1\right] - 2\Pr\left[X_u=1, X_v = -1 \right]}{2Pr\left[X_u=1,X_v = 1\right] + 2\Pr\left[X_u=1, X_v = -1 \right]} \label{eq:divbyone}\\
& ~~= \frac{Pr\left[X_u=1,X_v = 1\right] - \Pr\left[X_u=1, X_v = -1 \right]}{Pr\left[X_u=1,X_v = 1\right] + \Pr\left[X_u=1, X_v = -1 \right]} \\
& ~~= \left(\frac{\exp(2\th_{uv})-1}{\exp(2\th_{uv})+1}\right)\frac{\Pr\left[X_u=1, X_v = -1 \right]}{\Pr\left[X_u=1, X_v = -1 \right]} \label{eq:cor}\\
& ~~= \tanh(\th_{uv})
\end{align}
where \eqref{eq:sym} follows because $\Pr\left[X_u=1,X_v=1 \right] = \Pr\left[X_u=-1,X_v=-1 \right]$ and $\Pr\left[X_u=-1,X_v=1 \right] = \Pr\left[X_u=1,X_v=-1 \right]$ by symmetry. Line~\eqref{eq:divbyone} divides the expression by the total probability which is $1$ and \eqref{eq:cor} follows from Corollary~\ref{cor:trees}.

\end{prevproof}

\end{document}